\documentclass[12pt,reqno]{amsart}

\newcommand\version{June 18, 2017}

\usepackage{amsmath,amsfonts,amsthm,amssymb,amsxtra}
\usepackage{dsfont}

\newtheorem{theorem}{Theorem}
\newtheorem{proposition}[theorem]{Proposition}
\newtheorem{lemma}[theorem]{Lemma}
\newtheorem{corollary}[theorem]{Corollary}

\theoremstyle{definition}

\newtheorem{assumption}[theorem]{Assumption}

\theoremstyle{remark}

\newtheorem{remark}[theorem]{Remark}

\newcommand{\A}{\mathbf{A}}

\newcommand{\bB}{\mathbf{B}} 

\newcommand{\C}{\mathbb{C}}

\newcommand{\ch}{\mathord{\mathfrak h}}

\renewcommand{\epsilon}{\varepsilon}

\newcommand{\ii}{{\rm i}}

\newcommand{\N}{\mathbb{N}}

\renewcommand{\phi}{\varphi}
\newcommand{\R}{\mathbb{R}}

\newcommand{\Sph}{\mathbb{S}}

\newcommand{\Z}{\mathbb{Z}}
\newcommand\1{{\ensuremath {\mathds 1} }}

\DeclareMathOperator{\im}{Im}
\DeclareMathOperator{\ran}{ran}
\DeclareMathOperator{\re}{Re}
\DeclareMathOperator{\spec}{spec}

\begin{document}

\title[Critical temperature --- \version]{The BCS critical temperature in a weak homogeneous magnetic field}

\author[R. L. Frank]{Rupert L. Frank}
\address[R. L. Frank]{Mathematisches Institut der Universit\"at M\"unchen, Theresienstr. 39, 80333 M\"unchen, Germany, and Mathematics 253-37, Caltech, Pasadena, CA 91125, USA}
\email{rlfrank@caltech.edu}

\author[C. Hainzl]{Christian Hainzl}
\address[C. Hainzl]{Mathematisches Institut, Universit\"at T\"ubingen, Auf der Morgenstelle 10, 72076 T\"ubingen, Germany}
\email{christian.hainzl@uni-tuebingen.de}

\author[E. Langmann]{Edwin Langmann}
\address[E. Langmann]{Department of Physics, Royal Institute of Technology KTH, 106 91 Stockholm, Sweden}
\email{langmann@kth.se}

\begin{abstract}
We show that, within a linear approximation of BCS theory, a weak homogeneous magnetic field lowers the critical temperature by an explicit constant times the field strength, up to higher order terms. This provides a rigorous derivation and generalization of results obtained in the physics literature from WHH theory of the upper critical magnetic field. A new ingredient in our proof is a rigorous phase approximation to control the effects of the magnetic field.
\end{abstract}

\maketitle

\renewcommand{\thefootnote}{${}$} \footnotetext{\copyright\, 2017 by
  the authors. This paper may be reproduced, in its entirety, for
  non-commercial purposes.}

\section{Introduction and main result}

\subsection{Aims and scope}

In this paper we are interested in the infimum of the spectrum of the two-particle operator
\begin{equation}
\label{eq:twobodyk}
\frac{(-\ii\nabla_x + \frac12 \bB\wedge x)^2 + (-\ii\nabla_y + \frac12 \bB\wedge y)^2 - 2\mu}{\tanh \left( \frac\beta 2 \left( (-\ii\nabla_x + \frac12 \bB\wedge x)^2 - \mu \right)\right) + \tanh \left( \frac\beta 2 \left( (-\ii\nabla_y + \frac12 \bB\wedge y)^2 - \mu \right) \right)} - V(x-y)
\end{equation}
acting in
$$
L_{\rm symm}^2(\R^3\times\R^3) = \left\{\alpha\in L^2(\R^3\times\R^3):\ \alpha(x,y)=\alpha(y,x) \ \text{for all}\ x,y\in\R^3 \right\}.
$$
Here $-2V(x-y)$ is the interaction potential between the two particles which we assume to be spherically symmetric, i.e., it only depends on the distance $|x-y|$. (Later on, we will assume that the interaction potential is non-positive and the minus sign, as opposed to the more usual plus sign, will simplify some formulas.) Moreover, $\mu\in\R$ is the chemical potential. We are interested in the dependence of the operator on two parameters, namely, the inverse temperature $\beta>0$ and a constant magnetic field $\bB\in\R^3$, whose strength $B=|\bB|$ we shall assume to be small.

More precisely, we are interested in identifying regimes of temperatures $T=\beta^{-1}$ such that the infimum of the spectrum of the above operator is positive or negative for all sufficiently small $B$.

As we will explain in detail below, the motivation for this question comes from BCS theory of superconductivity and the operator arises through the linearization of the Bogolubov--de Gennes equation around the normal state. Therefore, the question whether the infimum of the spectrum of the operator \eqref{eq:twobodyk} is positive or negative corresponds to the local stability of the normal state. 

The largest magnetic field strength $B$ at a given temperature $T$ below the critical temperature $T_c$ where the normal state remains unstable is known in the physics literature as \emph{upper critical magnetic field} $B_{c2}(T)$, and it was first computed in the physics literature by Werthammer, Helfand and Hohenberg (WHH) based on an ansatz and certain simplifications \cite{HeWe,WeHeHo}; this is explained in Appendix \ref{appendix}. Our work provides a rigorous derivation of $B_{c2}(T)$ close to $T_c$ without these simplifications. Besides its mathematical interest, this is partly motivated by recently discovered superconducting materials challenging assumptions used in standard WHH theory.

To describe our main result we introduce the effective one-body operator
\begin{equation}
\label{eq:onebody}
\frac{(-\ii\nabla_r)^2 - \mu}{\tanh \left( \frac\beta 2 \left( (-\ii\nabla_r)^2 - \mu \right)\right)} - V(r)
\end{equation}
acting in
$$
L^2_{\rm symm}(\R^3) = \{ \alpha \in L^2(\R^3):\ \alpha(-r)=\alpha(r) \ \text{for all}\ r\in\R^3\} \,.
$$
Later on, we will see that the variable $r\in\R^3$ arises as the relative coordinate $r=x-y$ of the two particles at $x$ and $y$. We will \emph{assume} that the operator $|(-\ii\nabla_r)^2-\mu| -V(r)$ has a negative eigenvalue. Then it is easy to see (see, e.g., \cite{HHSS}) that there is a unique $\beta_c\in (0,+\infty)$ such that the operator \eqref{eq:onebody} is non-negative for $\beta\leq\beta_c$ and has a negative eigenvalue for $\beta>\beta_c$. Let $T_c = \beta_c^{-1}$. Then our main result is, roughly speaking, that the infimum of the spectrum of the two-particle operator \eqref{eq:twobodyk} is negative for $T \leq T_c - c_0 B + o(B)$ and positive for $T\geq T_c - c_0 B + o(B)$. Here $c_0$ is a positive constant which we compute explicitly in terms of the zero-energy ground state of \eqref{eq:onebody} at $\beta=\beta_c$.

The interpretation of this result is that, at least in linear approximation, a weak magnetic field lowers the BCS critical temperature to $T_c(B)=T_c-c_0 B+o(B)$. This is well-known in the physics literature \cite{HeWe,WeHeHo}, and we provide a rigorous mathematical proof without simplifying assumptions and with precise error bounds. In particular, our result proves that the slope of the upper critical field at the critical temperature
$$
\frac{d B_{c2}(T)}{dT}|_{T \uparrow T_c} := \lim_{B\to 0} \frac{B}{T_c(B)-T_c}
$$
is well-defined and equal to $-c_0^{-1}$. As discussed in Appendix \ref{appendix}, this reduces to the known result for this slope in WHH theory \cite{La,La2} in a limiting case.

The mathematical challenge of this problem is that low energy states of the two particle operator \eqref{eq:twobodyk} show a two-scale structure. As function of the relative coordinate $r=x-y$ and the center of mass coordinate $X=(x+y)/2$ it varies on a scale of order one with respect to $r$ and on a (much larger) scale of order $B^{-1}$ with respect to $X$. The variation on the former scale is responsible for the leading order term $T_c$ for the critical temperature, whereas the variation on the latter scale is responsible for the subleading lowering of order $B$. A similar separation of scales is typical for BCS theory near the critical temperature \cite{FHSS} and its effect on the critical temperature was explored in \cite{FHSS2}. We explain the main differences with \cite{FHSS2} after having presented our main result in the next subsection.

It will be more convenient for us to work not directly with the above two-particle operator, but rather with its Birman--Schwinger version. We now describe the precise set-up of our analysis.


\subsection{Model and main result}

Our model depends on the following ingredients.

\begin{assumption}\label{ass2}
(1) Homogeneous magnetic field $\bB=Be_3$ of strength $B>0$ in the direction of the third coordinate axis $e_3=(0,0,1)^T$\\
(2) Inverse temperature $\beta = T^{-1}>0$\\
(3) Chemical potential $\mu\in\R$\\
(4) Non-negative, spherically symmetric interaction potential $V$ such that $V\in L^\infty(\R^3)$ and $|r|V\in L^\infty(\R^3)$
\end{assumption}

Probably our analysis can be extended to cover some (not too severe) local singularities of $V$, but our boundedness assumptions allow us to avoid the related technicalities. Moreover, the non-negativity assumption on $V$ is only for technical convenience and we expect that our results hold true also for non-positive or sign-changing potentials satisfying the remaining assumptions.

Standard superconducting materials known at the time when WHH theory was developed are metallic in the normal state with $\mu>0$. We also allow for $\mu\leq 0$ since this is relevant for low-density superconductors like SrTiO$_3$ and for systems close to a superconductor-insulator phase transition.

We note that, in part (4) of Assumption \ref{ass2} $V$ refers to a function $\R^3\to\R$, $x\mapsto V(x)$. To simplify notation and since the precise meaning is always clear from the context, we use the same symbol $V$ also for the corresponding multiplication operators on $L^2_{\rm symm}(\R^3)$ (i.e., $(V\alpha)(r) = V (r)\alpha(r)$) and on $L^2_{\rm symm}(\R^3\times\R^3)$ (i.e., $(V\alpha)(x, y) = V (x-y)\alpha(x, y)$).

The magnetic momentum and the single-particle Hamiltonian are defined, respectively, by
$$
\pi = -\ii\nabla + \A  \qquad \text{with} \qquad
\A(x) = \frac12 \bB\wedge x = (B/2) (-x_2,x_1,0)^T
$$
and
$$
\ch_B = \pi^2 -\mu \,.
$$
The two particles are represented by coordinates $x,y\in\R^3$. If we want to emphasize the variables on which the operators act, we write
$$
\pi_x = -\ii\nabla_x + \frac12 \bB\wedge x \,,
\qquad
\pi_y = -\ii\nabla_y + \frac12 \bB\wedge y
$$
and
$$
\ch_{B,x} = \pi_x^2-\mu \,,
\qquad
\ch_{B,y} = \pi_y^2 -\mu \,.
$$
(Technically speaking these operators, as well as any operators in what follows, are considered as the Friedrichs extensions of the corresponding differential expressions acting on smooth and compactly supported functions.) Let us introduce a function $\Xi_\beta:\R^2\to\R$ by
$$
\Xi_\beta(E,E') := \frac{\tanh\frac{\beta E}{2}+\tanh\frac{\beta E'}{2}}{E+E'}
$$
if $E+E'\neq 0$ and $\Xi_\beta(E,-E) = (\beta/2)/\cosh^2(\beta E/2)$. (We comment in Remark \ref{notation} below on our non-standard notation.) Since the operators $\ch_{B,x}$ and $\ch_{B,y}$ commute, we can define the operator
$$
L_{T,B} = \Xi_\beta(\ch_{B,x},\ch_{B,y}) \,.
$$
We will always consider this operator in the Hilbert space $L^2_{\rm symm}(\R^3\times\R^3)$. Note that, with this notation, the operator in \eqref{eq:twobodyk} can be written as $L_{T,B}^{-1}+V$.

Next, in order to formulate our assumption on the critical temperature, we introduce the function $\chi_\beta:\R\to\R$ by
$$
\chi_\beta(E) := \frac{\tanh\frac{\beta E}2}{E}
$$
and set $\chi_\infty(E):=|E|^{-1}$. We consider the compact operator
$$
V^{1/2} \chi_\beta(p_r^2-\mu) V^{1/2}
$$
in $L^2_{\rm symm}(\R^3)$, where
$$
p_r = -\ii\nabla_r
$$
denotes the momentum operator. (The operator $\chi_\beta(p_r^2-\mu)$ is denoted by $K_T^{-1}$ in \cite{HHSS} and several works thereafter.) 

\begin{assumption}\label{ass0}
$\sup\spec V^{1/2}\chi_\infty(p_r^2-\mu)V^{1/2}>1$.
\end{assumption}

Since $\beta\mapsto\chi_\beta(E)$ is strictly increasing for each fixed $E\in\R$, there is a unique $\beta_c\in (0,\infty)$ such that
\begin{align*}
\sup \spec V^{1/2} \chi_\beta(p_r^2-\mu) V^{1/2} & \leq 1 \qquad\text{if}\ \beta\leq\beta_c \,,\\
\sup \spec V^{1/2} \chi_\beta(p_r^2-\mu) V^{1/2} & > 1 \qquad\text{if}\ \beta>\beta_c \,.
\end{align*}
We set $T_c=\beta_c^{-1}$.

\begin{assumption}\label{ass1}
The eigenvalue $1$ of the operator $V^{1/2} \chi_{\beta_c}(p_r^2-\mu) V^{1/2}$ is simple.
\end{assumption}

We denote by $\phi_*$ a normalized eigenfunction of $V^{1/2} \chi_\beta(p_r^2-\mu) V^{1/2}$ corresponding to the eigenvalue $1$ which, by assumption, is unique up to a phase. Since $p_r^2$ and $V$ are real operators, so is $V^{1/2} \chi_\beta(p_r^2-\mu) V^{1/2}$ and we can assume that $\phi_*$ is real-valued.

The spherical symmetry of $V$ from Assumption \ref{ass2} and the non-degeneracy from Assumption \ref{ass1} imply that $\phi_*$ is spherically symmetric.

From a physics point of view, Assumption \ref{ass1} restricts us to potentials giving raise to s-wave superconductivity. It is known that this assumption is fulfilled for a large class of potentials, including those which have a non-negative Fourier transform \cite{HaSe}.

As the final preliminary before stating our main result, we will introduce some constants. They are defined in terms of the auxiliary functions
\begin{align}\label{eq:auxiliary}
g_0(z) & = \frac{\tanh(z/2)}{z} \,,\notag \\
g_1(z) & = \frac{e^{2z}-2ze^z-1}{z^2(e^z+1)^2} = \frac1{2z^2} \frac{\sinh z-z}{\cosh^2(z/2)}\,, \notag \\
g_2(z) & = \frac{2e^z (e^z-1)}{z(e^z+1)^3} = \frac1{2z} \frac{\tanh(z/2)}{\cosh^2(z/2)}\,,
\end{align}
as well as the function
\begin{equation}
\label{eq:t}
t(p) := \|\chi_{\beta_c}((-\ii\nabla_r)^2-\mu) V^{1/2}\phi_*\|^{-1}\  2 (2\pi)^{-3/2} \int_{\R^3} dx\, V(x)^{1/2} \phi_*(x) e^{-\ii p\cdot x} \,.
\end{equation}
(The prefactor in front of the integral is irrelevant for us and only introduced for consistency with the definition in \cite{FHSS2}.) We now set
\begin{align}\label{eq:glcoeff}
\Lambda_0 & := \frac{\beta_c^2}{16} \int_{\R^3} \frac{dp}{(2\pi)^3}\, |t(p)|^2 \left(g_1(\beta_c(p^2-\mu)) + \frac23 \beta_c p^2 g_2(\beta_c(p^2-\mu)) \right) \,,\\
\Lambda_2 & := \frac{\beta_c}8 \int_{\R^3} \frac{dp}{(2\pi)^3}\, |t(p)|^2 \cosh^{-2}(\beta_c(p^2-\mu)/2) \,.
\end{align}
Note that the quotient $\Lambda_0/\Lambda_2$, which will appear in our main result, has the dimension of an inverse temperature.

We are now in position to state our main result.

\begin{theorem}\label{main}
Under assumptions \ref{ass2}, \ref{ass0} and \ref{ass1} the following holds.
\begin{enumerate}
\item[(1)] Let $0<T_1<T_c$. Then there are constants $B_0>0$ and $C>0$ such that for all $0< B\leq B_0$ and all $T_1\leq T< T_c - 2T_c(\Lambda_0/\Lambda_2)B -C B^2$ one has
$$
\inf_{\Phi} \langle\Phi, (1- V^{1/2} L_{T,B} V^{1/2})\Phi\rangle <0 \,.
$$
\item[(2)] There are constants $B_0>0$ and $C>0$ such that for all $0< B \leq B_0$ and all $T\geq T_c - 2T_c(\Lambda_0/\Lambda_2)B + C B^{5/4}$ one has
$$
\langle\Phi, (1- V^{1/2} L_{T,B} V^{1/2})\Phi\rangle >0 \,,
$$
unless $\Phi=0$.
\end{enumerate}
\end{theorem}

The interpretation of this theorem is that for small magnetic fields $B$ the critical temperature goes down by an amount $2T_c(\Lambda_0/\Lambda_2)B$ plus higher order terms. This gives, in particular, the slope of the upper critical field at $T_c$; see also Appendix \ref{appendix}.

The assumption in part (1) that the temperature is bounded away from zero is probably technical. Note however, that our result is valid for arbitrarily small $T_1>0$, as long as it is uniform in $B$. The reason for this restriction is that our expansions diverge as the temperature goes to zero. Remarkably, there is no such restriction in part (2) of the theorem.

\begin{remark}
Let us emphasize that our definition of critical temperature coincides with that in \cite{HHSS} (and therefore with that in \cite{FHSS,FHSS2}) and that our Assumptions \ref{ass0} and \ref{ass1} coincides with \cite[Assumption 2]{FHSS}. This is a consequence of the Birman--Schwinger principle, which also implies that, if $\alpha_*$ denotes a normalized, real-valued eigenfunction of the operator \eqref{eq:onebody}, then
$$
V^{1/2}\alpha_* = \pm \|\chi_{\beta_c}((\ii\nabla_r)^2-\mu) V^{1/2}\phi_*\|^{-1} \phi_* \,.
$$
(To get the normalization constant, we apply $\chi_{\beta_c}((-\ii\nabla_r)^2-\mu) V^{1/2}$ to both sides and use the equation for $\alpha_*$ and its normalization.)
\end{remark}

\begin{remark}\label{notation}
Our notation deviates somewhat from the standard one in the physics literature. Our $V$ corresponds to $-V/2$ in the physics literature. Here the factor of $1/2$, but not the minus sign, is consistent with \cite{HHSS}. The minus sign is used in order to simplify the formulas. Moreover, our $\chi_\beta$ and $\Xi_\beta$ correspond to $2\chi_\beta$ and $2\Xi_\beta$ in the physics literature. The factor $2$ here compensates the factor $1/2$ in $V$ in expressions like $V^{1/2}\chi_\beta(p^2-\mu)V^{1/2}$.
\end{remark}

Let us compare our results here with those in \cite{FHSS2} where we also computed the shift of the critical temperature due to external fields. The results of \cite{FHSS2}, which rely on those in \cite{FHSS} are more complete since they consider the physically relevant setting of a finite sample and, more importantly, since they treat the critical temperature in a non-linear setting (although eventually, it is proved that the critical temperature is determined by the linearization). It would be desirable to extend our results here in this direction and we believe that our analysis is the first and crucial step in this direction.

The reason why the same problem in our setting here of a homogeneous magnetic field is more complicated is the following. In \cite{FHSS,FHSS2} the external magnetic field was assumed to be periodic and to have flux zero through the boundaries of an elementary cell which, in some sense, means that it is a small perturbation. In particular, we could prove a priori bounds \cite[Lemmas 2 and 3]{FHSS} which do not contain the magnetic magnetic field neither in the relative nor in the center of mass variable. We do not expect a similar result to hold in our setting. Instead we prove an a priori bound which contains a magnetic field in the center of mass variable. This is essentially the content of Proposition \ref{opineq} and Lemma \ref{rbound}. The non-commutativity of the components of the magnetic momentum here leads to significant technical difficulties. The second novelty in this paper, as compared to \cite{FHSS} and \cite{FHSS2}, is the absence of semi-classical expansions and its replacement by the so-called phase approximation. This technique is well-known in the physics literature (it is used, for instance, in \cite{HeWe} in a related context) and appeared in the mathematics literature, for instance, in \cite{CoNe,Ne}. We feel that this technique is both conceptually and technically simpler and may have many applications in related problems.

\begin{remark}
Let us rewrite for a moment the magnetic field strength as $B=h^2$, such that $h$ denotes the ratio between the microscopic and the
macroscopic scale set by the field strength. Then our result can be stated in the following way. 
The constant magnetic field ${\bf B}$ lowers the critical temperature by  
$$ h^2 T_c D_c + o(h^2),$$ 
where 
$$ D_c := \frac {\Lambda_0} {\Lambda_2} \inf {\rm spec} (-\ii \nabla_X + e_3 \wedge X)^2 $$
can be interpreted as the lowest eigenvalue of the linearized Ginzburg--Landau operator. This formulation shows that our present result extends the earlier result \cite[Theorem 2.4]{FHSS2} to the case of constant magnetic fields. It further reproves the fact that the macroscopic fluctuations are captured by Ginzburg--Landau theory with parameters which are determined by the underlying microscopic system (that is, by $V$ and $\mu$). As we mentioned before, while in \cite{FHSS, FHSS2} the (magnetic) Laplace operator of the Ginzburg--Landau equation was recovered by tedious semi-classical expansions, in the present work this operator is simply recovered by changing to center of mass variables and a corresponding Taylor expansion, i.e., by expanding the cosine in \eqref{eq:repr2} below up to second order, which comes from a simple magnetic shift in the center of mass direction. 
\end{remark}


\subsection{Connection to BCS theory}

In this subsection we describe how the two-body operators \eqref{eq:twobodyk} and $L_{T,B}$ arise in a problem in superconductivity. Our purpose here is to give a motivation and our presentation in this subsection will be informal. For background and references on the mathematical study of BCS theory we refer to \cite{HSreview}.

We consider a superconducting sample occupying all of $\R^3$ at inverse temperature $\beta>0$ and chemical potential $\mu\in\R$. The particles interact through a two-body potential $-V(x-y)$ and are placed in an external magnetic field with vector potential $\A(x)$. In BCS theory the state of a system is described by two operators $\gamma$ and $\alpha$ in $L^2(\R^3)$, representing the one-body density matrix and the Cooper pair wave function, respectively. The operator $\gamma$ is assumed to be Hermitian and the operator $\alpha$ is assumed to satisfy $\alpha^* =\overline{\alpha}$, where for a general operator $A$ we write $\overline A = \mathcal C A \mathcal C$ with $\mathcal C$ denoting complex conjugation. Moreover, it is assumed that
$$
0\leq \begin{pmatrix}
\gamma & \alpha \\ \overline\alpha & 1-\overline\gamma
\end{pmatrix} \leq 1 \,.
$$

In an equilibrium state the operators $\gamma$ and $\alpha$ satisfy the (non-linear) Bogolubov--de Gennes equations
\begin{align*}
& \begin{pmatrix}
\gamma & \alpha \\ \overline\alpha & 1-\overline\gamma
\end{pmatrix} 
= \left( 1+ \exp\left( \beta H_{\Delta_{V,\alpha}}\right) \right)^{-1} \,, \\
& \qquad\text{where}\qquad
\Delta_{V,\alpha}(x,y) = -2V(x-y)\alpha(x,y)
\qquad\text{and}\qquad
H_{\Delta} = 
\begin{pmatrix}
\ch & \Delta \\ \overline\Delta & -\overline\ch
\end{pmatrix} \,.
\end{align*}
Here $\Delta$ is considered as an integral operator with integral kernel $\Delta(x,y)$. Moreover, $\ch =(-\ii\nabla +\A)^2-\mu$ is the one-particle operator.

Note that one solution of the equation is $\gamma= (1+\exp(\beta\ch))^{-1}$ and $\alpha=0$. This is the \emph{normal state}. We are interested in the local stability of this solution and therefore will linearize the equation around it.

It is somewhat more convenient to write the equation in the equivalent form
$$
\begin{pmatrix}
\gamma & \alpha \\ \overline\alpha & 1-\overline\gamma
\end{pmatrix} 
= \frac12 - \frac12 \tanh\left( \frac{\beta}2 H_{\Delta_{V,\alpha}}\right) \,.
$$
Then, in view of the partial fraction expansion (also known as Mittag--Leffler series)
$$
\tanh z = \sum_{n\in\Z} \frac{1}{z -\ii (n+1/2)\pi}
$$
(where we write $\sum_{n\in\Z}$ short for $\lim_{N\to\infty} \sum_{n=-N}^N$ for conditionally convergent sums like this one; convergence becomes manifest by combining the $+n$ and $-n$ terms),
$$
\tanh\left( \frac{\beta}2 H_{\Delta}\right) = - \frac2\beta \sum_{n\in\Z} \frac{1}{\ii\omega_n - H_\Delta}
$$
with the \emph{Matsubara frequencies}
\begin{equation}
\label{eq:matsubara}
\omega_n = \pi(2n+1)T \,, \qquad n\in\Z \,.
\end{equation}
Using this formula we can expand the operator $\tanh(\beta H_\Delta/2)$ in powers of $\Delta$. Since
\begin{align*}
\frac{1}{\ii\omega_n - H_\Delta} & = \frac{1}{\ii\omega_n - H_0} + \frac{1}{\ii\omega_n - H_0} \begin{pmatrix} 0 & \Delta \\ \overline\Delta & 0 \end{pmatrix} \frac{1}{\ii\omega_n - H_0} + \ldots \\
& = \begin{pmatrix} (\ii\omega_n - \ch)^{-1} & 0 \\ 0 & (\ii\omega_n + \overline\ch)^{-1} \end{pmatrix} \\
& \qquad + \begin{pmatrix} 0 & (\ii\omega_n - \ch)^{-1}\Delta (\ii\omega_n + \overline\ch)^{-1} \\ (\ii\omega_n + \overline\ch)^{-1}\overline \Delta (\ii\omega_n - \ch)^{-1} & 0\end{pmatrix} + \ldots \,,
\end{align*}
the Bogolubov--de Gennes equation for the Cooper pair wave function becomes
$$
\alpha = \frac1\beta \sum_{n\in\Z} (\ii\omega_n - \ch)^{-1}\Delta_{V,\alpha} (\ii\omega_n + \overline\ch)^{-1} + \ldots \,,
$$
where $\ldots$ stands for terms that are higher order in $\alpha$. The key observation now is that
\begin{equation}
\label{eq:ltsum}
\frac1\beta \sum_{n\in\Z} (\ii\omega_n - \ch)^{-1}\Delta_{V,\alpha} (\ii\omega_n + \overline\ch)^{-1} = L_{T,B} V\alpha \,.
\end{equation}
(Here $V\alpha$ on the right side is considered as a two-particle wave function, defined by $(V\alpha)(x,y)=V(x-y)\alpha(x,y)$.) This identity follows by writing
\begin{equation}
\label{eq:zetaidentity}
- \frac2\beta \sum_{n\in\Z} (\ii\omega_n - E)^{-1} (\ii\omega_n + E')^{-1} = 
- \frac{2}{\beta} \sum_{n\in\Z} \frac{1}{E+E'} \left( \frac{1}{\ii\omega_n - E} - \frac{1}{\ii\omega_n + E'} \right)
\end{equation}
and using the partial fraction expansion of $\tanh$ to recognize the right side as $\Xi_\beta(E,E')$.

Thus, the linearized Bogolubov--de Gennes equation becomes
$$
\alpha = L_{T,B} V\alpha \,.
$$
There are two ways to make this equation self-adjoint. The first one is to apply the operator $L_{T,B}^{-1}$ to both sides and to subtract $V\alpha$. In this way we obtain the operator \eqref{eq:twobodyk}. The other way is to multiply both sides of the equation by $V^{1/2}$, to subtract $V^{1/2} L_{T,B} V\alpha$ and to call $\Phi=V^{1/2}\alpha$. In this way we arrive at the operator $1-V^{1/2} L_{T,B} V^{1/2}$ which appears in our main result, Theorem \ref{main}.

The upshot of this discussion is that positivity of the operator \eqref{eq:twobodyk} (or, equivalently, of $1-V^{1/2} L_{T,B} V^{1/2})$ corresponds to local stability of the normal state and negativity of these operators corresponds to local instability. If we define two critical local temperatures $\overline{T_c^{\rm loc}(B)}$ as the smallest temperature above which the normal state is always stable and $\underline{T_c^{\rm loc}(B)}$ as the largest temperature below which the normal state is never stable, then our theorems says that (ignoring the presence of $T_1$ for simplicity) both $\overline{T_c^{\rm loc}(B)}$ and $\underline{T_c^{\rm loc}(B)}$ are equal to $T_c -c_0 B + o(B)$ as $B\to 0$ where $c_0 = 2T_c \Lambda_0/\Lambda_2>0$.


\subsection*{Acknowledgements}

We are very grateful to Michael Loss, whose ideas played a crucial role in finding the proof of Lemma \ref{decomp1}. E. L. would like to thank Yaron Kadem for helpful discussions. Partial support by the U.S. National Science Foundation through grant DMS-1363432 (R.L.F.) and by Vetenskapsr\aa det through grant 2016-05167 (E.L.) is acknowledged.


\section{Bounds on the resolvent kernel}

In this section we prove bounds on the resolvent kernel
\begin{equation}
\label{eq:reskernel}
G_B^{z}(x,y) := \frac 1{ z - \ch_B} (x,y)
\end{equation}
of the operator $\ch_B=\pi^2-\mu$ in $L^2(\R^3)$ with a constant magnetic field $\bB=B e_3$ with $B\geq 0$ and a chemical potential $\mu\in\R$. (Here and in the following, we use the convention that for an integral operator $K$ its integral kernel is denoted by $K(x,y)$.) We introduce the function
\begin{equation}
\label{eq:defgb}
g_B^z(x) := G_B^{z}(x,0) \,,
\qquad x\in\R^3 \,.
\end{equation}

We first collect some simple properties of this function.

\begin{lemma}\label{propgbz}
The function $g_B^z$ satisfies for all $B\geq 0$, $z\in\C\setminus[B,\infty)$ and $x,y\in\R^3$,
\begin{itemize}
\item[(i)] $g_B^z(-x)=g_B^z(x)$
\item[(ii)] $G_B^z(x,y) = e^{\frac \ii 2 \bB \cdot (x\wedge y)} g_B^z(x-y)$
\end{itemize}
\end{lemma}

\begin{proof}
We introduce coordinates $x=(x_\bot,x_3)$, $y=(y_\bot,y_3)$, perform a Fourier transform in the $x_\bot$ and $y_\bot$ variables and use the known structure of the spectrum of the Landau Hamiltonian. Thus, in terms of the projections $P_B^{(k)}$ in $L^2(\R^2)$ on the $k$-th Landau level, the kernel of $(z-\ch_B)^{-1}$ can be written as
$$
G_B^z(x,y) = \sum_{k\in\N_0} \int_{\R} \frac{dp_3}{2\pi} \frac{1}{z-(2k+1)B-p_3^2+\mu} e^{\ii p_3(x_3-y_3)} P_B^{(k)}(x_\bot,y_\bot) \,.
$$
Explicitly, the kernel of the projection $P_B^{(k)}$ is given by
$$
P_B^{(k)}(x_\bot,y_\bot) = \frac{B}{2\pi} L_k\left(\frac{B(x_\bot-y_\bot)^2}{2}\right) e^{-B(x_\bot-y_\bot)^2/4} e^{B\frac{\ii}{2}(x_1y_2-x_2y_1)} \,,
$$
where $L_k=L_k^{(0)}$ is the $k$-th Laguerre polynomial. From these formulas it is easy to deduce (i) and (ii).
\end{proof}

Our next goal is to quantify the decay of the $L^1$-norm of $g_B^z$ as $|z|\to\infty$ along the imaginary axis. We begin with the simpler case $B=0$. We employ the notation
\begin{equation}
\label{eq:posnegpart}
\mu_+ = \max\{\mu,0\} \,,
\qquad
\mu_- = -\min\{\mu,0\}\,,
\end{equation}
that is, $\mu=\mu_+-\mu_-$.

\begin{lemma}\label{freeres}
For every $a>-2$ there is a constant $C_a>0$ such that for all $\omega\in \R$ one has
$$
\left\||\cdot|^a g_0^{\ii\omega} \right\|_1 \leq C_a \left( \frac{|\omega|+\mu_+}{|\omega|(|\omega|+\mu_-)} \right)^{(a+2)/2} \,.
$$
\end{lemma}

\begin{proof}
One has
$$
g_0^z(x) = -\frac{e^{\ii\sqrt{z+\mu}|x|}}{4\pi|x|}
$$
with $\im\sqrt{z+\mu}\geq 0$ and therefore, by scaling,
\begin{align*}
\left\| |\cdot|^a g_0^z \right\|_1 & = \frac1{4\pi} \int_{\R^3} dx\, |x|^{a-1} e^{-\im\sqrt{z+\mu}|x|} = (\im\sqrt{z+\mu})^{-a-2} \frac1{4\pi} \int_{\R^3} dy\, |y|^{a-1} e^{-|y|} \\
& = (\im\sqrt{z+\mu})^{-a-2} \Gamma(a+2) \,.
\end{align*}
Since
$$
(\im\sqrt{\ii\omega+\mu})^2 = \frac{\sqrt{\mu^2+\omega^2}-\mu}{2}
\geq
\begin{cases}
\frac14 \frac{\omega^2}{|\omega|+\mu} & \text{if}\ \mu> 0 \,,\\
\frac12 (|\omega|+|\mu|) & \text{if}\ \mu\leq 0 \,,
\end{cases}
$$
we obtain the bound in the lemma.
\end{proof}
 
The next lemma deals with $B\neq 0$ by comparing it to the case $B=0$.
 
\begin{lemma}\label{magres}
There are constants $\delta>0$ and $C>0$ such that for all $B\geq 0$ and $\omega\in\R$ with $B^2(|\omega|+\mu_+)^2 \leq \delta \omega^2(|\omega|+\mu_-)^2$ one has
$$
\| g_B^{\ii\omega} - g_0^{\ii\omega} \|_1 \leq C B^2 \left( \frac{|\omega|+\mu_+}{|\omega| (|\omega|+\mu_-)} \right)^3 \,.
$$
\end{lemma} 

\begin{proof}
\emph{Step 1.} Let
$$
h^z(x) := \frac 14 (e_3 \wedge x)^2 g^z_0(x) \,.
$$
We claim that $g_B^z$ satisfies the equation
\begin{equation}
\label{eq:gbeq}
g_0^z(x) = g_B^z(x) - B^2 \int_{\R^3}dy\, e^{\frac \ii 2 \bB\cdot(x \wedge y)} g_B^z(x-y) h^z(y) \,.
\end{equation}
To see this, we note that, in the sense of operators,
$$
(z - \ch_B) G_0^z = \1 - T^z_B
$$
with
$$
T^z_B(x,y) := - (\bB \wedge x) \cdot \ii \nabla_x G^z_0(x-y) + \frac 14 (\bB \wedge x)^2 G^z_0(x-y) \,.
$$
Applying the operator $(z-\ch_B)^{-1}$ from the left we obtain
$$
G_0^z(x,y) = G_B^z(x,y) - \int_{\R^3} dw\, G_B^z(x,w) T_B^z(w,y)
$$
and, if we set $y=0$ and recall Lemma \ref{propgbz} (ii),
$$
g_0^z(x) = g_B^z(x) - \int_{\R^3} dw\, e^{\frac \ii 2 \bB\cdot(x \wedge w)} g_B^z(x-w) T_B^z(w,0) \,.
$$
To obtain the claimed equation it suffices to note that
$$
T^z_B(x,0) =  - (\bB \wedge x) \cdot \ii \nabla_x g^z_0(x) + \frac 14 (\bB \wedge x)^2 g^z_0(x) =\frac 14 (\bB \wedge x)^2 g^z_0(x) = B^2 h^z(x) \,,
$$
where we used the fact that the vector $\nabla_x g^z_0(x)$ is orthogonal to the vector $\bB \wedge x$.

\emph{Step 2.} We claim that, if $B^2<\|h^z\|_1^{-1}$, then
\begin{equation}
\label{eq:magresdiff}
\| g_B^z - g_0^z \|_1 \leq B^2  \frac{\| h^z\|_1 \|g^z_0\|_1}{1 - B^2 \|h^z\|_1} \,.
\end{equation}
In fact, \eqref{eq:gbeq} implies
$$
\| g_B^z - g_0^z \|_1 \leq B^2 \||g_B^z | \ast |h^z| \|_1 \leq B^2 \|g_B^z\|_1 \|h^z\|_1 \leq B^2 \|g_B^z - g^z_0\|_1 \|h^z\|_1 + B^2 \|g_0^z\|_1 \|h^z\|_1 \,.
$$
Here we denoted convolution by $*$ and we made use of Young's convolution inequality. This proves the claimed inequality.

\emph{Step 3.} We now conclude the proof of the proposition. We first observe that by the spherical symmetry of $g_0^z$
$$
\|h^z\|_1 = \frac23 \left\| |\cdot|^2 g_0^z \right\|_1 \,,
$$
and therefore, by Lemma \ref{freeres}, there is a $C>0$ such that for all $\omega\in\R$,
$$
\| h^{\ii\omega} \|_1 \leq C \left( \frac{|\omega|+\mu_+}{|\omega|(|\omega|+\mu_-)} \right)^2 \,.
$$
Thus, if $B^2\leq (1/(2C))|\omega|^2(|\omega|+\mu_-)^2/(|\omega|+\mu_+)^2$, then $B^2 \|h^{\ii\omega}\|_1\leq 1/2$ and therefore, by \eqref{eq:magresdiff},
$$
\| g_B^{\ii\omega} - g_0^{\ii\omega} \|_1 \leq 2 B^2 \|h^{\ii\omega}\|_1 \|g_0^{\ii\omega}\|_1 \,.
$$
Using once more the bound above on $\|h^{\ii\omega}\|_1$ as well as the bound from Lemma \ref{freeres} on $\|g_0^{\ii\omega}\|_1$ we obtain the bound claimed in the lemma.
\end{proof}


\section{A representation formula for the operator $L_{T,B}$}

In this section we derive a useful representation formula for the operator $L_{T,B}$ as a sum over contributions from the individual Matsubara frequencies $\omega_n$ from \eqref{eq:matsubara}. Moreover, we express the formula in terms of center of mass and relative coordinates,
$$
r = x-y \,,
\qquad
X= (x+y)/2 \,.
$$
The magnetic momentum in the center of mass coordinate is
\begin{equation}
\label{eq:pix}
\Pi_X = -\ii\nabla_X + 2\A(X) = -\ii\nabla_X + \bB\wedge X \,.
\end{equation}

\begin{lemma}\label{repr}
The operator $L_{T,B}$ acts as
$$
\left( L_{T,B}\alpha\right)(X+\frac r2,X-\frac r2) = \iint_{\R^3\times\R^3} dZ\,ds\, k_{T,B}(Z,r,s) \left(\cos(Z\cdot\Pi_X)\alpha\right)(X+\frac s2,X-\frac s2)
$$
with
$$
k_{T,B}(Z,r,s) := -\frac 2\beta \sum_{n\in\Z} k_{T,B}^n(Z,r,s)
$$
and
$$
k_{T,B}^n(Z,r,s) :=  g_B^{\ii\omega_n}(Z+ \frac{r-s}2)g_B^{-\ii\omega_n}(Z- \frac{r-s}2) e^{\frac{\ii}4 \bB\cdot(r\wedge s)} \,.
$$
\end{lemma}

\begin{proof}
Our starting point is \eqref{eq:ltsum} which, in terms of the resolvent kernel $G_B^z$ from \eqref{eq:reskernel}, implies that
\begin{align*}
\left( L_{T,B} \alpha \right)(x,y) 
& = -\frac2\beta \sum_{n\in\Z} \left( \frac 1{\ii\omega_n - \ch_B} \alpha \frac 1{\ii\omega_n+ \overline{\ch_B}}\right) (x,y) \\
& = - \frac2\beta \sum_{n\in\Z} \iint_{\R^3\times\R^3} G_B^{\ii\omega_n}(x,x') G_B^{-\ii\omega_n}(y,y') \alpha(x',y')\,dx'\,dy' \,.
\end{align*}
Here we used the fact that
$$
\frac 1{\ii\omega_n+ \overline{\ch_B}}(y',y) = G_B^{-\ii\omega_n}(y,y') \,,
$$
which follows from
\begin{equation*}
\frac 1{z + \overline{\ch_B}} = \overline{ \left( \frac 1{z + \ch_B}  \right)^\ast} \,.
\end{equation*}

According to Lemma \ref{propgbz} (ii), the full resolvent kernel $G_B^z$ can be recovered from $g_B^z$ and, after changing coordinates $X=(x+x')/2$, $r=x-x'$, $Y=(y+y')/2$ and $s=y-y'$, we obtain 
\begin{align*}
& \left(L_{T,B} \alpha\right)(X+\frac r2,X-\frac r2)
= -\frac2\beta \sum_{n\in\Z} \iint_{\R^3\times\R^3} dY ds \, \alpha(Y+\frac s2,Y-\frac s2) \\
& \qquad\qquad\qquad\qquad \times g_B^{\ii\omega_n}(X-Y+\frac{r-s}{2}) g_B^{-\ii\omega_n}(X-Y - \frac{r-s}{2}) e^{\ii  \bB\cdot(X \wedge Y)} e^{\frac \ii 4 \bB  \cdot( r\wedge s) } \,.
\end{align*}
Changing coordinates $Z=X-Y$, using $\ii\bB\cdot (X\wedge(X-Z))=-\ii Z\cdot (B\wedge X)$ and recalling the definition of $k_{T,B}(Z,r,s)$ we can write this as
$$
\left(L_{T,B} \alpha\right)(X+\frac r2,X-\frac r2)
= \iint_{\R^3\times\R^3} dZ ds \, k_{T,B}(Z,r,s)e^{-\ii Z \cdot ( \bB \wedge X)} \alpha(X-Z+\frac s2,X-Z-\frac s2) \,.
$$
Next, we use the fact that
$$
\psi(X - Z) = (e^{-\ii Z\cdot p_X} \psi)(X), \quad  p_X = -\ii\nabla_X \,.
$$
We recall definition \eqref{eq:pix} of $\Pi_X$ and note that $p_X=-\ii\nabla_X$ commutes with $\bB \wedge X$. Therefore we obtain
$$ 
e^{-  \ii Z \cdot(\bB \wedge X)} (e^{-\ii Z\cdot p_X} \psi)(X) 
= (e^{-\ii Z \cdot(p_X + \bB\wedge X)} \psi)(X)
= (e^{-\ii Z\cdot\Pi_X} \psi)(X) \,.
$$
This shows that
\begin{align}
\label{eq:reprproof}
\left(L_{T,B} \alpha\right)(X+\frac r2,X-\frac r2)
= \iint_{\R^3\times\R^3} dZ ds \, k_{T,B}(Z,r,s) \left( e^{-\ii Z\cdot\Pi_X} \alpha\right)(X+\frac s2,X-\frac s2) \,.
\end{align}
This almost proves the result, except that we still need to replace $e^{-\ii Z\cdot\Pi_X}$ by $\cos(Z\cdot\Pi_X)$. To do so, we change variables $Z\mapsto-Z$, $r\mapsto-r$ and $s\mapsto-s$ and use $\alpha(x,y)=\alpha(y,x)$ and $k_{T,B}(-Z,-r,-s)=k_{T,B}(Z,r,s)$ (see Lemma \ref{propgbz} (i)) in order to obtain the same formula as in \eqref{eq:reprproof}, but with $e^{-\ii Z\cdot\Pi_X}$ replaced by $e^{+\ii Z\cdot\Pi_X}$. Taking the mean of these two expressions proves the result.
\end{proof}

\begin{corollary}\label{reprproduct}
If $\alpha(X+r/2,X-r/2) = \psi(X)\tau(r)$ with $\tau$ even, then
\begin{align}
\label{eq:repr2}
\langle \alpha, L_{T,B} \alpha \rangle = \int_{\R^3} dZ\, \langle\psi,\cos(Z\cdot\Pi_X)\psi\rangle \iint_{\R^3\times\R^3} dr ds \, \overline{\tau(r)} k_{T,B}(Z,r,s) \tau(s) \,.
\end{align}
If, in addition, $\tau$ is real-valued, then
\begin{align*}
& \langle \alpha, L_{T,B} \alpha \rangle = \int_{\R^3} dZ\, \langle\psi,\cos(Z\cdot\Pi_X)\psi\rangle \\
& \times \!\left(-\frac 2\beta\right) \!\sum_n 
\iint_{\R^3\times\R^3} \!\!dr ds \, \tau(r) g_B^{\ii\omega_n}(Z+ \frac{r-s}2)g_B^{-\ii\omega_n}(Z- \frac{r-s}2) \cos(\frac{\ii}4 \bB\cdot(r\wedge s)) \tau(s) \,.
\end{align*}
\end{corollary}

In other words, the fact that $\tau$ is real-valued allows us to replace $e^{\frac{\ii}{4}\bB\cdot(r\wedge s)}$ by $\cos(\frac{\ii}4 \bB\cdot(r\wedge s))$. Since we will apply this in a regime where $B$ is small, this roughly corresponds to an improvement of the error from $B$ to $B^2$, which will be important for us.

\begin{proof}
The first formula follows immediately from Lemma \ref{repr}. To prove the second formula, we interchange the variables $r$ and $s$ and at the same time let $Z\mapsto-Z$. According to Lemma \ref{propgbz} (i) we have $k_{T,B}(-Z,s,r)=k_{T,B}(Z,r,s)e^{-\frac{\ii}2 \bB\cdot(r\wedge s)}$, and therefore \eqref{eq:repr2} becomes
$$
\langle \alpha, L_{T,B} \alpha \rangle = \int_{\R^3} dZ\, \langle\psi,\cos(Z\cdot\Pi_X)\psi\rangle \iint_{\R^3\times\R^3} dr ds \, \overline{\tau(s)} k_{T,B}(Z,r,s)e^{-\frac{\ii}2 \bB\cdot(r\wedge s)} \tau(r) \,.
$$
When $\tau$ is real, we can add this formula and \eqref{eq:repr2} and obtain the second formula in the corollary.
\end{proof}

We conclude this section with a formula which we need later and which, essentially, is the limiting case $B=0$ of Lemma \ref{repr}. We define, similarly as for $B>0$,
\begin{equation}
\label{eq:kt0zrs}
k_{T,0}(Z,r,s) := -\frac 2\beta \sum_{n\in\Z} k_{T,0}^n(Z,r,s)
\end{equation}
and
$$
k_{T,0}^n(Z,r,s) :=  g_0^{\ii\omega_n}(Z+ \frac{r-s}2)g_0^{-\ii\omega_n}(Z- \frac{r-s}2) \,.
$$
Then, setting $\rho=r-s$, $\ell=p+q$ and $k=(p-q)/2$ and recalling \eqref{eq:ltsum} and \eqref{eq:zetaidentity}, we obtain
\begin{align}
\label{eq:semest1}
k_{T,0}(X,r,s) & = - \frac2\beta \sum_{n\in\Z} \iint_{\R^3\times\R^3} \frac{dp}{(2\pi)^3}\,\frac{dq}{(2\pi)^3} \frac{e^{\ii p\cdot(Z+\frac\rho2)}}{\ii\omega_n-p^2+\mu} \frac{e^{\ii q\cdot(Z-\frac\rho2)}}{\ii\omega_n+q^2-\mu} \notag \\
& = - \iint_{\R^3\times\R^3} \frac{dp}{(2\pi)^3}\,\frac{dq}{(2\pi)^3} L(p,q) e^{\ii p\cdot(Z+\frac\rho2)+\ii q\cdot(Z-\frac\rho2)} \notag \\
& = - \iint_{\R^3\times\R^3} \frac{d\ell}{(2\pi)^3}\,\frac{dk}{(2\pi)^3} L(k+\frac\ell2,k-\frac\ell2) e^{\ii\ell\cdot Z +\ii k\cdot\rho}
\end{align}
with
\begin{equation}
\label{eq:semestl}
L(p,q) : = \frac{\tanh\frac{\beta(p^2-\mu)}2+\tanh\frac{\beta(q^2-\mu)}2}{p^2-\mu+q^2-\mu} \,.
\end{equation}



\section{Approximation of the operator $L_{T,B}$}

This section contains the technical heart of this paper. We shall approximate the operator $L_{T,B}$ by increasingly simpler operators. Namely, we shall write
$$
L_{T,B} = \left( L_{T,B}-M_{T,B}\right) + \left( M_{T,B} - N_{T,B} \right) + N_{T,B}
$$
with certain operators $M_{T,B}$ and $N_{T,B}$ and in Subsection \ref{sec:approxl} we shall show that both differences in parentheses are small when $B$ is small. In the following subsection we investigate in more detail the operator $N_{T,B}$ and show that a leading order approximation for small $B$ is $\chi_\beta(p_r^2-\mu)$. Then we proceed to find the subleading correction, which will be the key for proving our main result. These approximations are based on a method which we explain in Subsection \ref{sec:method}.

In this section we keep precisely track of the exact parameter dependence of the error terms, even if we do not need this in the present paper. We do this in order to emphasize the explicitness of our method, which might be applicable in different limiting regimes as well.

One of the important technical novelties in this paper compared to \cite{FHSS} is the treatment of the magnetic field via the phase approximation; see the introduction for references. This appears in Subsection \ref{sec:approxl} and we show that this approximation is valid provided $B\beta(1 + \beta\mu_+)\leq\delta(1+\beta\mu_-)$, where $\delta$ is a small dimensionless constant.

Throughout this section Assumption \ref{ass2} is in effect.


\subsection{The method}\label{sec:method}

The following proposition is our main technical tool in order to perform the phase approximation.

\begin{proposition}\label{abstractbound}
Let $\ell$ be a measurable function on $\R^3\times\R^3\times\R^3$ such that
\begin{align}
\label{eq:abstractkernel1}
|\ell(Z,r,s)| \leq C_1 \sum_{n\in\Z} & \left( g_1^{(n)}(Z+\frac{r-s}{2}) g_2^{(n)}(Z-\frac{r-s}{2}) \right. \notag \\
& \qquad \left. + g_3^{(n)}(Z+\frac{r-s}{2})g_4^{(n)}(Z-\frac{r-s}{2}) \right)
\end{align}
with functions $g_1^{(n)},\ldots,g_4^{(n)}\in L^1(\R^3)$ satisfying
\begin{equation}
\label{eq:abstractkernel2}
\| g_1^{(n)} \|_1 \| g_2^{(n)} \|_1 +\| g_3^{(n)} \|_1 \| g_4^{(n)} \|_1 \leq C_2 \left( \frac{|2n+1|+\nu_+}{|2n+1|\left( |2n+1|+\nu_-\right)} \right)^a 
\end{equation}
for some $a>1$ and $\nu\in\R$. Moreover, let $A_{Z,r,s}$, $Z,r,s\in\R^3$, be a measurable family of bounded operators on $L^2(\R^3)$ such that
\begin{equation}
\label{eq:abstractkernel3}
\sup_{Z,r,s\in\R^3} \| A_{Z,r,s} \|\leq C_3 \,.
\end{equation}
Then the operator $\mathcal L$ in $L^2(\R^3\times\R^3)$ defined by
$$
(\mathcal L\alpha)(X+\frac{r}{2},X-\frac{r}{2}) = \iint_{\R^3\times\R^3} dZ\,ds\, \ell(Z,r,s) (A_{Z,r,s}\alpha)(X+\frac{s}{2},X-\frac{s}{2})
$$
(where $A_{Z,r,s}$ acts on the center of mass variable $X$) is bounded and
$$
\|\mathcal L\| \leq C_a C_1 C_2 C_3 \frac{(1+\nu_+)^a}{(1+\nu_-)^{a-1}} \,.
$$
\end{proposition}

We recall that $\nu_\pm$ denote the positive and negative parts of $\nu$. (In fact, the proposition remains true if $\nu_+$ and $\nu_-$ are two arbitrary non-negative numbers, not necessarily arising as positive and negative parts of a common $\nu$.)

The proof will be based on the following simple boundedness criterion.

\begin{lemma}\label{abstractlemma}
Let $B_{Z,r,s}$, $Z,r,s\in\R^3$, be a measurable family of bounded operators in $L^2(\R^3)$ such that
$$
C:= \left( \sup_{r\in\R^3} \iint_{\R^3\times\R^3} dZ\,ds\, \|B_{Z,r,s}\| \right)^{1/2} \left( \sup_{s\in\R^3} \iint_{\R^3\times\R^3} dZ\,dr\, \|B_{Z,r,s}\| \right)^{1/2} <\infty
$$
and defined the operator $\mathcal B$ in $L^2(\R^3\times\R^3)$ by
$$
(\mathcal B\alpha)(X+\frac{r}{2},X-\frac{r}{2}) = \iint_{\R^3\times\R^3} dZ\,ds\, (B_{Z,r,s}\alpha)(X+\frac{s}{2},X-\frac{s}{2}) \,,
$$
where the operators $B_{Z,r,s}$ act on the center of mass variable $X$. Then $\mathcal B$ is bounded with $\|\mathcal B\|\leq C$.
\end{lemma}

\begin{proof}[Proof of Lemma \ref{abstractlemma}]
By Minkowski's integral inequality we have for each fixed $r\in\R^3$
\begin{align*}
& \left( \int_{\R^3} dX\, \left| (\mathcal B\alpha)(X+\frac{r}{2},X-\frac{r}{2}) \right|^2 \right)^{1/2} \\
& \qquad \leq \iint_{\R^3\times\R^3} dZ\,ds \left( \int_{\R^3} dX\, \left| (B_{Z,r,s}\alpha)(X+\frac{s}{2},X-\frac{s}{2}) \right|^2 \right)^{1/2} \\
& \qquad \leq \iint_{\R^3\times\R^3} dZ\,ds \|B_{Z,r,s}\| M(s)^{1/2} \\
& \qquad = \int_{\R^3} ds\, b(r,s) M(s)^{1/2} \,,
\end{align*}
where we have set
$$
M(s) := \int_{\R^3} dX\, \left| \alpha(X+\frac{s}{2},X-\frac{s}{2}) \right|^2
$$
and
$$
b(r,s) := \int_{\R^3} dZ\, \|B_{Z,r,s}\| \,.
$$
The assumption $C<\infty$ implies, by the Schur test, that the operator $b$ in $L^2(\R^3)$ with kernel $b(r,s)$ is bounded with $\|b\|\leq C$. Therefore,
$$
\| \mathcal B \alpha \| \leq \|b M^{1/2} \| \leq C \|M^{1/2}\| = C \|\alpha\| \,,
$$
as claimed.
\end{proof}

\begin{proof}[Proof of Proposition \ref{abstractbound}]
We will apply Lemma \ref{abstractlemma} with
$$
B_{Z,r,s} = \ell(Z,r,s) A_{Z,r,s} \,.
$$
By assumption we have
\begin{align*}
\|B_{Z,r,s}\| \leq C_1 C_3 \sum_{n\in\Z} & \left( g_1^{(n)}(Z+\frac{r-s}{2}) g_2^{(n)}(Z-\frac{r-s}{2}) \right.\\
& \left. \qquad + g_3^{(n)}(Z+\frac{r-s}{2})g_4^{(n)}(Z-\frac{r-s}{2}) \right),
\end{align*}
and therefore
$$
\int_{\R^3} dZ\, \|B_{Z,r,s}\| \leq C_1 C_3 \sum_{n\in\Z} \left( g_1^{(n)}* \tilde g_2^{(n)}(r-s) + g_3^{(n)}*\tilde g_4^{(n)}(r-s) \right),
$$
where $\tilde g_j^{(n)}(r) := g_j^{(n)}(-r)$ and where $*$ denotes convolution. By Young's convolution inequality,
$$
\iint_{\R^3\times\R^3} dZ\,dr\, \|B_{Z,r,s}\| \leq C_1 C_3 \sum_{n\in\Z} \left( \| g_1^{(n)}\|_1 \| g_2^{(n)}\|_1 + \| g_3^{(n)}\|_1 \| g_4^{(n)}\|_1 \right)
$$
and similarly for $\iint_{\R^3\times\R^3} dZ\,ds\, \|B_{Z,r,s}\|$. We now insert the assumed bound on the $L^1$ norms of the $g_j^{(n)}$ and bound
\begin{align*}
\sum_{n\in\Z} \left( \frac{|2n+1|+\nu_+}{|2n+1|\left( |2n+1|+\nu_-\right)}\right)^a
& = 2 \sum_{n=0}^\infty \left( \frac{2n+1+\nu_+}{(2n+1)\left( 2n+1+\nu_-\right)}\right)^a \\
& \leq 2 \sum_{n=0}^\infty \left( \frac{1+\nu_+}{2n+1+\nu_-}\right)^a.
\end{align*}
Thus, the claimed inequality will follow from the bound
$$
2 \sum_{n=0}^\infty \frac{1}{(2n+1+\nu_-)^a} \leq \frac{C_a}{(1+\nu_-)^{a-1}} \,,
$$
which can be shown by an easy comparison with the corresponding integral.
\end{proof}


\subsection{Approximation of the operator $L_{T,B}$}\label{sec:approxl}

Let us define an operator $M_{T,B}$ on $L^2_{\rm symm}(\R^3\times\R^3)$ by
$$
\left( M_{T,B}\alpha\right)(X+\frac r2,X-\frac r2) = \iint_{\R^3\times\R^3} dZ\,ds\, k_{T,B}^M(Z,r,s) \left(\cos(Z\cdot\Pi_X)\alpha\right)(X+\frac s2,X-\frac s2)
$$
with
$$
k_{T,B}^M(Z,r,s) := -\frac 2\beta \sum_n k_{T,B}^{n,M}(Z,r,s)
$$
and
$$
k_{T,B}^{n,M}(Z,r,s) :=  g_0^{\ii\omega_n}(Z+ \frac{r-s}2)g_0^{-\ii\omega_n}(Z- \frac{r-s}2) e^{\frac{\ii}4 \bB\cdot(r\wedge s)} \,.
$$
The difference between this operator and the operator $L_{T,B}$ is that $g_B^z$ is replaced by $g_0^z$. We show that the operators are close when $B$ is small.

\begin{lemma}\label{lltilde}
There are $\delta>0$ and $C>0$ such that for all $\beta>0$ and $B>0$ with $B\beta(1 + \beta\mu_+)\leq\delta( 1+ \beta\mu_-)$ one has
$$
\left\| \left( L_{T,B} - M_{T,B} \right)\alpha \right\| \leq C B^2 \beta^3 \frac{(1+\beta\mu_+)^4}{(1+\beta\mu_-)^3} \|\alpha\| \,.
$$
\end{lemma}

\begin{proof}
We write
$$
L_{T,B} - M_{T,B} = L_{T,B}^{(1)} + L_{T,B}^{(2)} \,,
$$
where the operators $L_{T,B}^{(1)}$ and $L_{T,B}^{(2)}$ are of the same form as $L_{T,B}$ and $M_{T,B}$, but with kernels given by
\begin{align*}
k_{T,B}^{n,1}(Z,r,s) &:=  \left( g_B^{\ii\omega_n} - g_0^{\ii\omega_n} \right)(Z+ \frac{r-s}2) g_0^{-\ii\omega_n}(Z- \frac{r-s}2) e^{\frac{\ii}4 \bB\cdot(r\wedge s)} \\
& \qquad + g_0^{\ii\omega_n}(Z+ \frac{r-s}2) \left( g_B^{-\ii\omega_n} - g_0^{-\ii\omega_n}\right)(Z- \frac{r-s}2) e^{\frac{\ii}4 \bB\cdot(r\wedge s)} \,,\\
k_{T,B}^{n,2}(Z,r,s) & :=  \left( g_B^{\ii\omega_n}- g_0^{\ii\omega_n}\right)(Z+ \frac{r-s}2) \left( g_B^{-\ii\omega_n} - g_0^{-\ii\omega_n}\right)(Z- \frac{r-s}2)  e^{\frac{\ii}4 \bB\cdot(r\wedge s)} \,.
\end{align*}
We claim that we are in the setting of Proposition \ref{abstractbound} with
\begin{equation}
\label{eq:achoice1}
A_{Z,r,s} = \cos(Z\cdot\Pi_X) \,,\
\text{so that}\ \| A_{Z,r,s}^{(n)}\|\leq 1=C_3 \,.
\end{equation}
Moreover, the kernel $\ell(Z,r,s)$ is pointwise bounded as in \eqref{eq:abstractkernel1} with $C_1=2/\beta$ and
$$
g_1^{(n)} = |g_B^{\ii\omega_n}-g_0^{\ii\omega_n}|\,,\ 
g_2^{(n)}=|g_0^{-\ii\omega_n}|\,,\
g_3^{(n)}=|g_0^{\ii\omega_n}|\,,\
g_4^{(n)}=|g_B^{-\ii\omega_n}-g_0^{-\ii\omega_n}| \,.
$$
The $L^1$ norms of $g_2$ and $g_3$ are bounded by Lemma \ref{freeres}. We want to bound the $L^1$ norms of $g_1$ and $g_4$ using Lemma \ref{magres} and, to do so, we need that the assumption $B^2(|\omega_n| + \mu_+)^2\leq\delta \omega_n^2(|\omega_n|+\mu_-)^2$ is satisfied for any $n\in\Z$, which is equivalent to $B^2(\pi T + \mu_+)^2 \leq \delta \pi^2 T^2(\pi T + \mu_-)^2$. This is implied by $\beta B(1+\beta\mu_+)\leq\delta' (1+\beta\mu_-)$ for a suitable $\delta'>0$. (Here and in all the following we estimate $1+\beta\mu_\pm \leq\pi+\beta\mu_\pm\leq \pi(1+\beta\mu_\pm)$ in order to obtain nicer expressions.) Under this assumption we therefore obtain \eqref{eq:abstractkernel2} with $a=4$, $\nu=\beta\mu$ and $C_2 = C B^2 \beta^4$. Thus, Proposition \ref{abstractbound} yields the bound
$$
\| L_{T,B}^{(1)} \| \leq C' B^2 \beta^3 \frac{(1+\beta\mu_+)^4}{(1+\beta\mu_-)^3} \,.
$$

The argument for the operator $L_{T,B}^{(2)}$ is similar with the same choice \eqref{eq:achoice1} for $A_{Z,r,s}$. Now \eqref{eq:abstractkernel1} holds with $C_1=2/\beta$ and
$$
g_1^{(n)} = |g_B^{\ii\omega_n}-g_0^{\ii\omega_n}| \,,\
g_2^{(n)} = |g_B^{-\ii\omega_n}-g_0^{-\ii\omega_n}| \,,\
g_3^{(n)} = g_4^{(n)} = 0 \,.
$$
As before, Lemma \ref{magres} yields \eqref{eq:abstractkernel2} with $a=6$, $\nu=\beta\mu$ and $C_2= C B^4\beta^6$ and therefore Proposition \ref{abstractbound} yields
$$
\| L_{T,B}^{(2)} \| \leq C' B^4 \beta^5 \frac{(1+\beta\mu_+)^6}{(1+\beta\mu_-)^5} \,.
$$
Since $B\beta(1+\beta\mu_+)\leq\delta'(1+\beta\mu_-)$, this is bounded by $C'\delta'^2 B^2 \beta^3 (1+\beta\mu_+)^4/(1+\beta\mu_-)^3$, which proves the lemma.
\end{proof}

Next, we define an operator $N_{T,B}$ on $L^2_{\rm symm}(\R^3\times\R^3)$ by
\begin{equation}
\label{eq:defmt}
\left( N_{T,B}\alpha\right)(X+\frac r2,X-\frac r2) := \iint_{\R^3\times\R^3} dZ\,ds\, k_{T,0}(Z,r,s) \left(\cos(Z\cdot\Pi_X)\alpha\right)(X+\frac s2,X-\frac s2)
\end{equation}
with $k_{T,0}(Z,r,s)$ from \eqref{eq:kt0zrs}. The difference between this operator and $M_{T,B}$ is that $k_{T,0}(Z,r,s)$, in contrast to $k_{T,B}^{n,M}(Z,r,s)$, does not depend on $B$.

\begin{lemma}\label{ltildem}
There is a $C>0$ such that for all $T>0$ and $B>0$,
$$
\left\| |r|^{-1/2} \left( M_{T,B} - N_{T,B} \right)\alpha \right\| \leq C B\beta^{3/2} \, \frac{(1+\beta\mu_+)^{5/2}}{(1+\beta\mu_-)^{3/2}}\, \||r|^{1/2} \alpha\| \,.
$$
\end{lemma}

\begin{proof}
We write
$$
M_{T,B} - N_{T,B} = L_{T,B}^{(3)} \,,
$$
where $L_{T,B}^{(3)}$ is of the same form as $L_{T,B}$ and $M_{T,B}$, but with kernels given by
\begin{align*}
k_{T,B}^{n,3}(Z,r,s) :=  g_0^{\ii\omega_n}(Z+ \frac{r-s}2)g_0^{-\ii\omega_n}(Z- \frac{r-s}2) \left( e^{\frac{\ii}4 \bB\cdot(r\wedge s)} - 1 \right).
\end{align*}
We will bound the operator $|r|^{-1/2}L_{T,B}^{(3)}|r|^{-1/2}$ using Proposition \ref{abstractbound}. We again make the choice \eqref{eq:achoice1} and set $\ell(Z,r,s) = -(2/\beta)\sum_n |r|^{-1/2}k_{T,B}^{n,3}(Z,r,s) |s|^{-1/2}$. Moreover, in order to bound the kernel we estimate
\begin{align*}
\left| e^{\frac{\ii}4 \bB\cdot(r\wedge s)} - 1 \right| & = 2 \left| \sin \left(\frac{\bB}8\cdot(r\wedge s) \right) \right| \leq \left| \frac{\bB}{4}\cdot(r\wedge s) \right| \\
& \leq \frac{B}{4} |r\wedge(r-s)|^{1/2} |(r-s)\wedge s|^{1/2} \leq \frac{B}{4} |r|^{1/2} |r-s| |s|^{1/2} \\
& \leq \frac{B}{4} |r|^{1/2} \left( \left| Z+\frac{r-s}{2}\right| + \left|Z- \frac{r-s}{2}\right|\right) |s|^{1/2} \,.
\end{align*}
Thus, we obtain the bound \eqref{eq:abstractkernel1} with $C_1=B/(2\beta)$ and
$$
g_1^{(n)} = |\cdot| |g_0^{\ii\omega_n}| \,,\
g_2^{(n)} = |g_0^{-\ii\omega_n}| \,,\
g_3^{(n)} = |g_0^{\ii\omega_n}| \,,\
g_4^{(n)} = |\cdot| |g_0^{-\ii\omega_n}| \,. 
$$
According to Lemma \ref{freeres} we have the bound \eqref{eq:abstractkernel2} with $a=5/2$, $\nu=\beta\mu$ and $C_2=C\beta^{5/2}$. Thus, Proposition \ref{abstractbound} yields
$$
\| |r|^{-1/2} L_{T,B}^{(3)} |r|^{-1/2} \| \leq C' B \beta^{3/2} \frac{(1+\beta\mu_+)^{5/2}}{(1+\beta\mu_-)^{3/2}} \,,
$$
which is the claimed bound.
\end{proof}

Lemma \ref{ltildem} yields, in particular, the bound
$$
\left|\langle\alpha,\left( M_{T,B} - N_{T,B} \right)\alpha \rangle \right| \leq C B \beta^{3/2} \frac{(1+\beta\mu_+)^{5/2}}{(1+\beta\mu_-)^{3/2}}\  \||r|^{1/2} \alpha\|^2 \,.
$$
The drawback of this bound is that the right side is \emph{linear} in $B$. We now show that for $\alpha$ of a special form we obtain a \emph{quadratic} bound.

\begin{lemma}\label{ltildemform}
If $\alpha(X+r/2,X-r/2)=\psi(X)\tau(r)$ with $\tau$ even and real-valued, then
$$
\left|\left\langle\alpha,\left( M_{T,B} - N_{T,B} \right)\alpha \right\rangle \right| \leq C B^2 \beta^2 \frac{( 1+\beta\mu_+)^3}{(1+\beta\mu_-)^2} \||r| \alpha\|^2 \,.
$$
\end{lemma}

\begin{proof}
Let us define an operator $\tilde M_{T,B}$ in $L^2_{\rm symm}(\R^3\times\R^3)$ which is of the same form as $M_{T,B}$, but with the factor $e^{\frac{\ii}{4}\bB\cdot(r\wedge s)}$ replaced by $\cos(\frac{1}{4}\bB\cdot(r\wedge s))$. By the same argument as in the proof of Corollary \ref{reprproduct} we have for $\alpha$ of the form in the lemma that
$$
\left\langle\alpha,M_{T,B}\alpha\right\rangle = \left\langle\alpha,\tilde M_{T,B}\alpha\right\rangle \,.
$$
Therefore, in order to prove the lemma, it suffices to bound the norm of the operator $|r|^{-1/2}(\tilde M_{T,B} - N_{T,B})|r|^{-1/2}$. We do with using Proposition \ref{abstractbound} and make again the choice \eqref{eq:achoice1} for $A_{Z,r,s}$. Moreover, we bound
$$
\left|\cos \left(\frac \bB 4 \cdot (r\wedge s)\right) -1\right| = 2\sin^2\left(\frac \bB 8   \cdot(r\wedge s) \right) \leq \frac{B^2}{32} |r\wedge s|^2
$$
and
\begin{align*}
|r\wedge s|^2 & = |r\wedge(r-s)||(r-s)\wedge s| \leq |r| |r-s|^2 |s| \\
& \leq |r| \left( \left| Z + \frac{r-s}{2} \right| + \left|Z-\frac{r-s}{2}\right|\right)^2 |s| \\
& \leq 2 |r| \left( \left| Z + \frac{r-s}{2} \right|^2 + \left|Z-\frac{r-s}{2}\right|^2\right) |s| \,.
\end{align*}
Thus, \eqref{eq:abstractkernel1} holds for $\ell(Z,r,s)= |r|^{-1/2} k_{T,0}(Z,r,s) |s|^{-1/2} (\cos(\frac{1}{4}\bB\cdot(r\wedge s)) - 1)$ with $C_1=B^2/(8\beta)$ and$$
g_1 = |\cdot|^2 |g_0^{\ii\omega_n}|\,,\
g_2 = |g_0^{-\ii\omega_n}| \,,\
g_3 = |g_0^{\ii\omega_n}| \,,\
g_4= |\cdot|^2 |g_0^{-\ii\omega_n}| \,.
$$
According to Lemma \ref{freeres} we have \eqref{eq:abstractkernel2} with $a=3$, $\nu=\beta\mu$ and $C_2=C\beta^3$. Therefore, Proposition \ref{abstractbound} yields
$$
\left\| |r|^{-1/2}(\tilde M_{T,B} - N_{T,B})|r|^{-1/2} \right\| \leq C'  B^2\beta^2 \frac{(1+\beta\mu_+)^3}{(1+\beta\mu_-)^2} \,,
$$
which is the claimed bound.
\end{proof}


\subsection{Approximation of the operator $N_{T,B}$}

Recall that the operator $N_{T,B}$ was defined in \eqref{eq:defmt}, and that its definition involves the operator $\cos(Z\cdot\Pi_X)$. In this subsection we approximate the operator $N_{T,B}$ first with an operator where the cosine is replaced by $1$, and then with an operator where it is replaced by $1-(1/2)(Z\cdot\Pi_X)^2$.

\begin{lemma}\label{mchi}
There is a constant $C>0$ such that for all $T>0$ and all $B>0$,
$$
\left\| \left( N_{T,B} - \chi_\beta(p_r^2-\mu) \right)\alpha \right\| \leq C\beta^2\ \frac{(1+\beta\mu_+)^3}{(1+\beta\mu_-)^2} \ \|\Pi_X^2\alpha\| \,.
$$
\end{lemma}

\begin{proof}
The key observation is the following expression for the operator $\chi_\beta(p_r^2-\mu)$,
\begin{equation}
\label{eq:reprchi}
\left( \chi_\beta(p_r^2-\mu) \alpha\right)(X+\frac r2,X-\frac r2) = \iint_{\R^3\times\R^3} dZ\,ds\, k_{T,0}(Z,r,s) \alpha(X + \frac s2,X-\frac s2) \,.
\end{equation}
Indeed, according to \eqref{eq:semest1} we have
\begin{align*}
\int_{\R^3}dZ\, k_{T,0}(Z,r,s) & = -\frac2\beta\sum_n \int_{\R^3} dZ\, k_{T,0}^n(Z,r,s) \\
& =\int_{\R^3} dZ\, \iint_{\R^3\times\R^3} \frac{dk}{(2\pi)^3}\frac{d\ell}{(2\pi)^3} L(k+\frac{\ell}{2},k-\frac{\ell}{2}) e^{\ii\ell\cdot Z + \ii k\cdot(r-s)}
\end{align*}
with $L(p,q)$ from \eqref{eq:semestl}. Doing the $Z$ and the $\ell$ integrations, we obtain
$$
\int_{\R^3} \frac{dk}{(2\pi)^3} \, L(k,k) e^{\ii k\cdot(r-s)} = 
\int_{\R^3} \frac{dk}{(2\pi)^3} \,\chi_\beta(k^2-\mu) e^{\ii k\cdot(r-s)} \,,
$$
which yields \eqref{eq:reprchi}.

Identity \eqref{eq:reprchi} allows us to write the operator $(N_{T,B}-\chi_\beta(p_r^2-\mu))(\Pi_X^2)^{-1}$ in the form of Proposition \ref{abstractbound} with the choice
$$
A_{Z,r,s} = |Z|^{-2} (\cos(Z\cdot\Pi_X) - 1)(\Pi_X^2)^{-2} \,.
$$
The inequality
\begin{align}
\label{eq:cosbound}
(1-\cos\lambda)^2 \leq \frac{1}{4} \lambda^4
\end{align}
implies that
$$
A_{Z,r,s}^* A_{Z,r,s} \leq \frac1{4|Z|^4} (\Pi_X^2)^{-2} (Z\cdot\Pi_X)^4  (\Pi_X^2)^{-2} \,.
$$
By repeated use of the Schwarz inequality it is easy to see that there is a constant $C$ such that for any self-adjoint operators $A_1,A_2,A_3$ and real scalars $\alpha_1,\alpha_2,\alpha_3$,
\begin{equation}
\label{eq:threeoperators}
\left(\alpha_1 A_1 + \alpha_2 A_2 + \alpha_3 A_3 \right)^4 \leq C^2 \left( \alpha_1^2 + \alpha_2^2 + \alpha_3^2\right)^2 (A_1^2 + A_2^2 + A_3^2)^2 \,.
\end{equation}
This implies that $A_{Z,r,s}^*A_{Z,r,s}\leq (C/2)^2$, that is,
$$
\| A_{Z,r,s} \| \leq C/2 = C_3 \,.
$$
Let us bound the kernel $\ell(Z,r,s)=Z^2 k_{T,0}(Z,r,s)$ pointwise. Using
\begin{equation}
\label{eq:boundz}
Z^2 \leq \frac{1}{4} \left( \left| Z + \frac{r-s}{2}\right| + \left| Z- \frac{r-s}{2}\right|\right)^2
\leq \frac12 \left( \left| Z + \frac{r-s}{2}\right|^2 + \left| Z- \frac{r-s}{2}\right|^2\right)
\end{equation}
we obtain \eqref{eq:abstractkernel1} with $C_1=1/\beta$ and
$$
g_1 = |\cdot|^2 |g_0^{\ii\omega_n}|\,,\
g_2 = |g_0^{-\ii\omega_n}| \,,\
g_3 = |g_0^{\ii\omega_n}| \,,\
g_4= |\cdot|^2 |g_0^{-\ii\omega_n}| \,.
$$
According to Lemma \ref{freeres} we have \eqref{eq:abstractkernel2} with $a=3$, $\nu=\beta\mu$ and $C_2 =C \beta^3$. Thus, Proposition \ref{abstractbound} yields
$$
\left\| \left( N_{T,B} - \chi_\beta(p_r^2-\mu) \right)(\Pi_X^2)^{-1} \right\| \leq C'\beta^2\ \frac{(1+\beta\mu_+)^3}{(1+\beta\mu_-)^2}\,,
$$
which is the claimed bound.
\end{proof}

The following lemma is somewhat technical. It plays a key role in removing a cut-off in the proof of the upper bound on the critical temperature and it is crucial to have a superlinear power of $\Pi_X^2$ in the norm on the right side.

\begin{lemma}\label{mchicomm}
There is a constant $C>0$ such that for all $T>0$ and all $B\geq 0$,
$$
\left\| |r|^{-1} \left( e^{\mp\ii\Pi_X\cdot r/2} - 1\right) \left( N_{T,B} - \chi_\beta(p_r^2-\mu) \right)\alpha \right\| \leq C \beta^2 \frac{( 1+ \beta\mu_+)^3}{(1+\beta\mu_-)^2} \left\|\left(\Pi_X^2\right)^{3/2}\alpha\right\| \,.
$$
\end{lemma}

\begin{proof}
Our starting point is again \eqref{eq:reprchi}, which allows us to write the operator 
$$
|r|^{-1} \left( e^{\mp\ii\Pi_X\cdot r/2} - 1\right) \left( N_{T,B} - \chi_\beta(p_r^2-\mu) \right) (\Pi_X^2)^{-3/2}
$$
in the form of Proposition \ref{abstractbound} with
$$
A_{Z,r,s} = |r|^{-1} |Z|^{-2} \left( e^{\mp\ii\Pi_X\cdot r/2} - 1\right) \left( 1- \cos(Z\cdot\Pi_X) \right) (\Pi_X^2)^{-3/2} \,.
$$
Using the inequality $|e^{\mp\ii\lambda/2} -1|^2 \leq\lambda^2/4$ we obtain
\begin{align*}
A_{Z,r,s}^*A_{Z,r,s} & \leq \frac1{4r^2|Z|^4} (\Pi_X^2)^{-3/2}\left( 1- \cos(Z\cdot\Pi_X) \right) (r\cdot\Pi_X)^2 \left( 1- \cos(Z\cdot\Pi_X) \right) (\Pi_X^2)^{-3/2} \\
& \leq \frac1{4|Z|^4} (\Pi_X^2)^{-3/2}\left( 1- \cos(Z\cdot\Pi_X) \right) \Pi_X^2 \left( 1- \cos(Z\cdot\Pi_X) \right) (\Pi_X^2)^{-3/2}
\end{align*}
We will prove momentarily that
\begin{equation}
\label{eq:mchicommproof}
\left(1-\cos(Z\cdot\Pi_X)\right) \Pi_X^2 \left(1-\cos(Z\cdot\Pi_X)\right) \leq C^2 |Z|^4 (\Pi_X^2)^3 \,.
\end{equation}
This implies that $A_{Z,r,s}^*A_{Z,r,s} \leq \frac{C^2}{4}$, that is, $\|A_{Z,r,s}\|\leq C/2$.

The kernel $\ell(Z,r,s)=Z^2 k_0(Z,r,s)$ has already been estimated in the proof of Lemma \ref{mchi}. Thus, we obtain by Proposition \ref{abstractbound} that
$$
\left\| |r|^{-1} \left( e^{\mp\ii\Pi_X\cdot r/2} - 1\right) \left( N_{T,B} - \chi_\beta(p_r^2-\mu) \right) (\Pi_X^2)^{-3/2} \right\| \leq C \beta^2 \frac{( 1+ \beta\mu_+)^3}{(1+\beta\mu_-)^2} \,,
$$
which is the claimed bound.

We are left to prove the estimate \eqref{eq:mchicommproof}. The first step in the proof is to rewrite the left side as
\begin{align}
\label{eq:mchicommproof1}
\left(1-\cos(Z\cdot\Pi_X)\right) \Pi_X^2 \left(1-\cos(Z\cdot\Pi_X)\right)
& = \sum_j \Pi_X^{(j)}\left( 1-\cos(Z\cdot\Pi_X)\right)^2 \Pi_X^{(j)} \notag \\
& \qquad - 4 B^2 Z_\bot^2 \cos(Z\cdot\Pi_X) \left(1-\cos(Z\cdot\Pi_X)\right),
\end{align}
where we use the notation $\Pi_X^{(1)}=-\ii\partial_{X_1}-BX_2$, $\Pi_X^{(2)}=-\ii\partial_{X_2}+BX_1$ and $\Pi_X^{(3)}=-\ii\partial_{X_3}$. We have
$$
\left[ \Pi_X^{(i)},\Pi^{(j)}_X\right] = -2\ii B \epsilon_{ij}
$$
with $\epsilon_{12}=1$, $\epsilon_{21}=-1$ and $\epsilon_{ij}=0$ otherwise. This implies
\begin{align*}
\left[\cos(Z\cdot\Pi_X),\Pi^{(j)}_X\right] & = 2\ii B \sum_i Z_i \epsilon_{ij} \sin(Z\cdot\Pi_X) \,,\\
\left[\sin(Z\cdot\Pi_X),\Pi^{(j)}_X\right] & = -2\ii B \sum_i Z_i \epsilon_{ij} \cos(Z\cdot\Pi_X) \,.
\end{align*}
Thus, we obtain
\begin{align*}
& \left(1-\cos(Z\cdot\Pi_X)\right) \left(\Pi_X^{(j)} \right)^2 \left(1-\cos(X\cdot\Pi_X)\right) \\
& = \left( \Pi_X^{(j)} \left(1-\cos(Z\cdot\Pi_X)\right) - 2\ii B \sum_i Z_i \epsilon_{ij} \sin(Z\cdot\Pi_X) \right) \\
& \qquad  \times \left( \left(1-\cos(Z\cdot\Pi_X)\right) \Pi_X^{(j)} + 2\ii B \sum_i Z_i \epsilon_{ij} \sin(Z\cdot\Pi_X) \right) \\
& = \Pi_X^{(j)} \left(1-\cos(Z\cdot\Pi_X)\right)^2 \Pi_X^{(j)} + 4B^2 \sin^2(Z\cdot\Pi_X) \sum_{i,i'} Z_i Z_{i'} \epsilon_{ij}\epsilon_{i'j} \\
& \qquad + 2\ii B \sum_i Z_i \epsilon_{ij} \left[ \Pi_X^{(j)} , \left(1-\cos(Z\cdot\Pi_X)\right) \sin(Z\cdot\Pi_X) \right].
\end{align*}
Since $\cos\lambda \sin\lambda = (1/2)\sin(2\lambda)$, we can rewrite the last term as
\begin{align*}
\left[ \Pi_X^{(j)} , \left(1-\cos(Z\cdot\Pi_X)\right) \sin(Z\cdot\Pi_X) \right]
& = \left[ \Pi_X^{(j)} , \sin(Z\cdot\Pi_X)- (1/2)\sin(2Z\cdot\Pi_X) \right] \\
& = 2\ii B \sum_i Z_i \epsilon_{ij} \left(\cos(Z\cdot\Pi_X) - \cos(2Z\cdot\Pi_X)\right).
\end{align*}
Finally, we sum over $j$ and use the fact that
$$
\sum_j \sum_{i,i'} Z_i Z_{i'} \epsilon_{ij}\epsilon_{i'j} = Z_\bot^2
$$
in order to obtain
\begin{align*}
& \sum_j \left(1-\cos(Z\cdot\Pi_X)\right) \left(\Pi_X^{(j)} \right)^2 \left(1-\cos(X\cdot\Pi_X)\right) = \sum_j \Pi_X^{(j)} \left(1-\cos(Z\cdot\Pi_X)\right)^2 \Pi_X^{(j)} \\
& \qquad + 4B^2 Z_\bot^2 \left( \sin^2(Z\cdot\Pi_X) - \cos(Z\cdot\Pi_X) + \cos(2Z\cdot\Pi_X)\right).
\end{align*}
Since $\sin^2\lambda + \cos(2\lambda) = \cos^2\lambda$, this is the same as \eqref{eq:mchicommproof1}.

We now bound the right side of \eqref{eq:mchicommproof1} from above. Recalling \eqref{eq:cosbound} and \eqref{eq:threeoperators} we have
$$
\sum_j \Pi_X^{(j)}\left( 1-\cos(Z\cdot\Pi_X)\right)^2 \Pi_X^{(j)} \leq \frac14 \sum_j \Pi_X^{(j)} (Z\cdot\Pi_X)^4 \Pi_X^{(j)} \leq \frac C4\ |Z|^4  \sum_j \Pi_X^{(j)} (\Pi_X^2)^2 \Pi_X^{(j)} \,.
$$
We compute and estimate, using $\Pi_X^2\geq 2B$,
$$
\sum_j \Pi_X^{(j)} (\Pi_X^2)^2 \Pi_X^{(j)} = \left(\Pi_X^2\right)^3 + 32 B^2 \Pi_X^2 \leq 9 \left(\Pi_X^2\right)^3 \,.
$$
Moreover, since $\cos\lambda (1-\cos\lambda)\geq -\lambda^2/2$,
\begin{align*}
B^2 Z_\bot^2 \cos(Z\cdot\Pi_X) \left(1-\cos(Z\cdot\Pi_X)\right) & \geq - \frac{B^2}{2} Z_\bot^2 (Z\cdot\Pi_X)^2 \geq - \frac{B^2}{2} Z_\bot^2 |Z|^2 \Pi_X^2 \\
& \geq  - \frac{1}{8}|Z|^4 (\Pi_X^2)^3 \,.
\end{align*}
This proves \eqref{eq:mchicommproof}.
\end{proof}

So far, in Lemmas \ref{mchi} and \ref{mchicomm} we have seen that $N_{T,B}$ is given to leading order by $\chi_\beta(p_r^2-\mu)$. We now extract the subleading term.

\begin{lemma}\label{lemrest}
There is a constant $C>0$ such that for all $T>0$, all $B\geq 0$ and all $\alpha$ of the form $\alpha(X+r/2,X-r/2)=\psi(X)\tau(r)$ with $\tau$ radially symmetric and real-valued, one has
\begin{align*}
& \left| \langle\alpha,N_{T,B}\alpha\rangle -   \int_{\R^3} d Z\, F_\tau(Z) \|\psi\|^2  + \frac 16 \int_{\R^3} d Z\, Z^2 F_\tau(Z)  \langle \psi, \Pi_X^2 \psi\rangle \right| \\
& \qquad \leq C \beta^3 \ \frac{(1+\beta\mu_+)^4}{(1+\beta\mu_-)^3}\
\|\Pi_X^2\alpha\|^2 \,,
\end{align*}
where
$$
F_\tau(Z) := \iint_{\R^3\times\R^3} drds\, \tau(r) k_{T,0}(Z,r,s) \tau(s) \,.
$$
\end{lemma}

\begin{proof}
Let us introduce an operator $O_{T,B}$ in $L^2_{\rm symm}(\R^3\times\R^3)$ by
\begin{align*}
& \left( O_{T,B}\alpha\right)(X+\frac r2,X-\frac r2) \\
& := \iint_{\R^3\times\R^3} dZ\,ds\, k_{T,0}(Z,r,s) \left(\left(\cos(Z\cdot\Pi_X)-1+\frac12(Z\cdot\Pi_X)^2\right)\alpha\right)(X+\frac s2,X-\frac s2) \,.
\end{align*}
We claim that
$$
\langle\alpha,N_{T,B}\alpha\rangle -   \int_{\R^3} d Z\, F_\tau(Z) \|\psi\|^2  + \frac 16 \int_{\R^3} d Z\, Z^2 F_\tau(Z)  \langle \psi, \Pi_X^2 \psi\rangle
= \langle\alpha,O_{T,B}\alpha\rangle \,,
$$
This is clear for the first two terms on the left side, which correspond to the terms $\cos(Z\cdot\Pi_X)$ and $-1$ in the definition of $O_{T,B}$. For the third term on the left side, which corresponds to the term $(1/2)(Z\cdot\Pi_X)^2$ on the right side, we use the fact that $Z\mapsto F_\tau(Z)$ is spherically symmetric (which easily follows from the spherical symmetry of $\tau$ and of $g_0^{\pm\ii\omega_n}$) to deduce that
\begin{equation}
\label{eq:pisquared}
\int_{\R^3} dZ\, F_\tau(Z) (Z\cdot\Pi_X)^2 = \frac13 \int_{\R^3} dZ\, F_\tau(Z) Z^2 \Pi_X^2 \,.
\end{equation}
In fact, this follows by multiplying out the left side and using the fact that the angular average of $Z_i Z_j$ is $(Z^2/3)\delta_{ij}$. This proves the claimed formula

Thus, it remains to bound the norm of the operator $(\Pi_X^2)^{-1}O_{T,B}(\Pi_X^2)^{-1}$. This follows again by Proposition \ref{abstractbound} with the choice
$$
A_{Z,r,s} = |Z|^{-4} (\Pi_X^2)^{-1} \left(\cos(Z\cdot\Pi_X)-1+\frac12(Z\cdot\Pi_X)^2\right) (\Pi_X^2)^{-1}
$$
and $\ell(Z,r,s)=|Z|^4k_{T,0}(Z,r,s)$. In order to bound $\|A_{Z,r,s}\|$ we use the fact that 
\begin{equation}
\label{eq:cosineq}
1-\frac12 x^2 \leq \cos x \leq 1-\frac12 x^2+\frac1{24}x^4\qquad\text{for all}\ x\in\mathbb{R} \,.
\end{equation}
Because of this inequality and \eqref{eq:threeoperators}
$$
A_{Z,r,s} \leq \frac1{24} |Z|^{-4} (\Pi_X^2)^{-1} (Z\cdot\Pi_X)^4 (\Pi_X^2)^{-1} \leq \frac{C}{24} \,.
$$
Similarly, one shows $A_{Z,r,s}\geq 0$ and therefore $\|A_{Z,r,s}\|\leq C/24=C_3$.

We bound $\ell$ pointwise using $|Z|^4 \leq (1/2) (|Z+(r-s)/2|^4 + |Z-(r-s)/2|^4)$. This leads to \eqref{eq:abstractkernel1} with $C_1=1/\beta$ and
$$
g_1 = |\cdot|^4 g_0^{\ii\omega_n} \,,\
g_2= g_0^{-\ii\omega_n} \,,\
g_3 = g_0^{\ii\omega_n} \,,\
g_4 = |\cdot|^4 g_0^{-\ii\omega_n} \,.
$$
Then, from Lemma \ref{freeres} we obtain \eqref{eq:abstractkernel3} with $a=4$, $\nu=\beta\mu$ and $C_2=C'\beta^4$.

Therefore, Proposition \ref{abstractbound} yields
$$
\| (\Pi_X^2)^{-1}O_{T,B}(\Pi_X^2)^{-1} \| \leq C'' \beta^3 \frac{(1+\beta\mu_+)^4}{(1+\beta\mu_-)^3} \,,
$$
which concludes the proof of the lemma.
\end{proof}


\section{Weak magnetic field estimates} 

We consider functions $\alpha\in L^2_{\rm symm}(\R^3\times\R^3)$ of the form
$$
\alpha(x,y) = \tau(x-y) \psi((x+y)/2)
$$
with $\tau\in L^2_{\rm symm}(\R^3)$ and $\psi\in L^2(\R^3)$. The following theorem computes the expectation value of $L_{T,B}$ in states of this form. The bound will turn into an asymptotic expansion in the case where $\tau$ varies on a shorter scale than $\psi$.

\begin{theorem} \label{semest}
There are constants $\delta>0$ and $C>0$ such that the following holds. If $B\beta(1+\beta\mu_+)\leq\delta(1+\beta\mu_-)$ and if $\alpha$ is of the form
$$
\alpha(X+r/2,X-r/2)=\psi(X)\tau(r)
$$
with $\tau$ spherically symmetric and real-valued, then
\begin{align}
\label{eq:semest}
& \left| \langle \alpha, L_{T,B} \alpha \rangle -  A^{(0)}_T[\tau] \|\psi\|^2 - A^{(1)}_T[\tau] \langle \psi, \Pi_X^2 \psi\rangle \right| \notag \\
& \qquad \leq C \left( \beta^3 \ \frac{(1+\beta\mu_+)^4}{(1+\beta\mu_-)^3}\ \|\tau\|^2 \|\Pi_X^2\psi\|^2 + B^2 \beta^2\ \frac{(1+\beta\mu_+)^3}{(1+\beta\mu_-)^2}\ \left\| |\cdot|\tau\right\|^2 \| \psi\|^2 \right)
\end{align}
with
\begin{align*}
A^{(0)}_T[\tau] & = \beta \int_{\R^3} dp\, |\hat \tau(p)|^2 \ g_0(\beta(p^2-\mu)) \,, \\
A^{(1)}_T[\tau] & = -\frac{\beta^2}{4} \int_{\R^3} dp\, |\hat \tau(p)|^2 \left(g_1(\beta(p^2-\mu)) + \frac23 \beta p^2 g_2(\beta(p^2-\mu)) \right)
\end{align*}
in terms of the functions $g_0$, $g_1$ and $g_2$ from \eqref{eq:auxiliary}.
\end{theorem}

\begin{proof}
Combining Lemmas \ref{lltilde}, \ref{ltildemform} and \ref{lemrest} we obtain
\begin{align*}
& \left| \langle\alpha,L_{T,B}\alpha\rangle -   \int_{\R^3} d Z\, F_\tau(Z) \|\psi\|^2  + \frac 16 \int_{\R^3} d Z\, Z^2 F_\tau(Z)  \langle \psi, \Pi_X^2 \psi\rangle \right| \\
& \qquad \leq C \left( \beta^3 \ \frac{(1+\beta\mu_+)^4}{(1+\beta\mu_-)^3}\ \|\tau\|^2 \|\Pi_X^2\psi\|^2 + B^2 \beta^2\ \frac{(1+\beta\mu_+)^3}{(1+\beta\mu_-)^2}\ \left\| |\cdot|\tau\right\|^2 \| \psi\|^2 \right).
\end{align*}
Here we have bounded $(2B)^2\|\psi\|^2 \leq \|\Pi_X^2\psi\|^2$ to simplify the form of the remainder.

Therefore it remains to show that
$$
A^{(0)}_T[\tau] = -\int_{\R^3} d Z\, F_\tau(Z)
\qquad\text{and}\qquad
A^{(1)}_T[\tau] = \frac 16 \int_{\R^3} d Z\, Z^2 F_\tau(Z) \,.
$$
To do so, we multiply identity \eqref{eq:semest1} by
$$
\int_{\R^3} dr\, \tau(r+\frac\rho2)\tau(r-\frac\rho2) = \int_{\R^3} dr\, \overline{\tau(r+\frac\rho2)}\tau(r-\frac\rho2) = \int_{\R^3} dp'\, |\hat \tau(p')|^2 e^{-ip'\cdot\rho}
$$
and integrate with respect to $\rho$ to get
\begin{align*}
F_\tau(Z) & = - \iint_{\R^3\times\R^3} d\rho dp' \int_{\R^3\times\R^3} \frac{d\ell}{(2\pi)^3}\,\frac{dk}{(2\pi)^3} L(k+\frac\ell2,k-\frac\ell2) e^{\ii\ell\cdot Z +\ii k\cdot\rho} |\hat \tau(p')|^2 e^{-ip'\cdot\rho} \\
& = - \int_{\R^3} dk \int_{\R^3} \frac{d\ell}{(2\pi)^3}\, L(k+\frac\ell2,k-\frac\ell2) e^{\ii\ell\cdot Z} |\hat \tau(k)|^2 \,.
\end{align*}
This implies
$$
\int_{\R^3} dZ\, F_\tau(Z) = - \int_{\R^3} dk \, L(k,k) |\hat \tau(k)|^2
$$
and
$$
\int_{\R^3} dZ\, Z^2 F_\tau(Z) = \int_{\R^3} dk \, \nabla_{\ell}^2|_{\ell=0} L(k+\frac\ell2,k-\frac\ell2) |\hat \tau(k)|^2 \,.
$$
Clearly,
$$
L(k,k) = \beta\, g_0(\beta(k^2-\mu)) \,,
$$
and a tedious, but straightforward computation yields
$$
\nabla_{\ell}^2|_{\ell=0} L(k+\frac\ell2,k-\frac\ell2) = - \frac{3\beta^2}2 \left(g_1(\beta(k^2-\mu)) + \frac23 \beta k^2 g_2(\beta(k^2-\mu)) \right)
$$
in terms of the functions $g_0$, $g_1$ and $g_2$ defined in \eqref{eq:auxiliary}. This finishes the proof of Theorem~\ref{semest}.
\end{proof}



\section{Lower bound on the critical temperature}

We now provide the \emph{Proof of part (1) of Theorem \ref{main}}, which will be a rather straightforward consequence of Theorem \ref{semest}. We will work under Assumptions \ref{ass2} and \ref{ass0}. Assumption \ref{ass1} is not needed in this part of Theorem \ref{main}.

We fix a parameter $T_1$ with $0<T_1<T_c$ and restrict ourselves to temperatures $T\geq T_1$. We consider functions $\Phi$ in $L^2_{\rm symm}(\R^3\times\R^3)$ of the form
$$
\Phi(x,y)=\phi(x-y)\psi((x+y)/2)\,,
$$
where the functions $\phi\in L^2_{\rm symm}(\R^3)$ and $\psi\in L^2(\R^3)$ are still to be determined. At the moment we require only that $\|\psi\|=1$, $\|\Pi_X^2\psi\|<\infty$ and $\||\cdot| \phi\|<\infty$. Applying the expansion from Theorem \ref{semest} with $\tau(r)=V(r)^{1/2}\phi(r)$ we find that, as long as $B\beta(1+\beta\mu_+)\leq\delta(1+\beta\mu_-)$,
\begin{align*}
\langle\Phi, (1- V^{1/2} L_{T,B} V^{1/2})\Phi\rangle = & \|\phi\|^2 - \langle\tau(r)\psi(X), L_{T,B} \tau(r)\psi(X) \rangle \\
\leq & \|\phi\|^2 - A_T^{(0)}[\tau] - A_T^{(1)}[\tau] \langle\psi,\Pi_X^2\psi\rangle + C_T \|\Pi_X^2\psi\|^2 \,,
\end{align*}
where
$$
C_T  = C \left( \beta^3 \ \frac{(1+\beta\mu_+)^4}{(1+\beta\mu_-)^3}\ \|V^{1/2}\phi\|^2 + \beta^2\ \frac{(1+\beta\mu_+)^3}{(1+\beta\mu_-)^2}\ \left\| |\cdot|V^{1/2}\phi\right\|^2 \right).
$$
We have $C_T<\infty$ by our assumptions on $\phi$ and the assumption that $V\in L^\infty(\R^3)$.

The leading order term on the right side is
$$
\|\phi\|^2 - A_T^{(0)}[\tau] = \left\langle \phi, \left( 1- V^{1/2}\chi_\beta(p_r^2-\mu) V^{1/2}\right)\phi\right\rangle
$$
Motivated by this expression we choose
$$
\phi=(2\pi)^{-3/2} \|\chi_{\beta_c}(p^2-\mu) V^{1/2} \phi_*\| \ \phi_*
$$
which makes this term equal to zero at $T=T_c$. (The prefactor here is irrelevant and only used to obtain the precise form of the coefficients $\Lambda_0$ and $\Lambda_2$. The quotient $\Lambda_0/\Lambda_2$ is independent of this choice of normalization.) Note that \cite[Proposition 1]{FHSS} guarantees that $\||\cdot| \phi\|<\infty$.

With this choice of $\phi$ we therefore obtain
\begin{align}
\label{eq:upperproof1}
\langle\Phi, (1- V^{1/2} L_{T,B} V^{1/2})\Phi\rangle 
\leq & A^{(0)}_{T_c}[\tau] - A^{(0)}_T[\tau] - A^{(1)}_T[\tau] \langle\psi,\Pi_X^2\psi\rangle + C_T \|\Pi_X^2\psi\|^2 \,.
\end{align}

In order to proceed, we note the fact that $\tau = V^{1/2}\phi = (2\pi)^{-3/2} V \alpha_*$, and therefore, in terms of the function $t$ from \eqref{eq:t},
\begin{equation}
\label{eq:upperchoicetau}
\hat\tau = (1/2)(2\pi)^{-3/2} t \,.
\end{equation}
It follows from this identity that
$$
\frac{d}{dT}|_{T=T_c} A^{(0)}_{T}[\tau] = T_c^{-1} \Lambda_2 \,,
$$
and some simple analysis of the function $g_0$ shows that
$$
A^{(0)}_{T_c}[\tau] - A^{(0)}_T[\tau] \leq -\Lambda_2 \frac{T_c-T}{T_c} + C'(T_c-T)^2
$$
for all $T_1\leq T\leq T_c$. Using \eqref{eq:upperchoicetau} once again we also find that
$$
A^{(1)}_{T_c}[\tau] = -\Lambda_0 \,,
$$
which in turn can be used to prove that
$$
A^{(1)}_T[\tau] \geq -\Lambda_0 - C''(T_c-T)
$$
for all $T_1\leq T\leq T_c$.

Inserting these expansions into \eqref{eq:upperproof1} we obtain
\begin{align}
\label{eq:upperproof2}
\langle\Phi, (1- V^{1/2} L_{T,B} V^{1/2})\Phi\rangle 
\leq & -\Lambda_2 \frac{T_c-T}{T_c} + \Lambda_0 \langle\psi,\Pi_X^2\psi\rangle \notag \\
& + C'(T_c-T)^2 + C''(T_c-T) \langle\psi,\Pi_X^2\psi\rangle + C_T \|\Pi_X^2\psi\|^2
\end{align}
for all $T_1\leq T\leq T_c$. We now choose $\psi$ in order to make the term $\langle\psi,\Pi_X^2\psi\rangle$ as small as possible (with $\|\psi\|=1$). To do so, we introduce coordinates $X=(X_\bot,X_3)$ with $X_\bot\in\R^2$ and $X_3\in\R$ and we define
$$
\psi(X) = \sqrt{2B} \psi_\bot(\sqrt{2B} X_\bot) \ell^{-1/2} \psi_{\parallel}(X_3/\ell) \,.
$$
Here $\psi_\bot$ is a normalized ground state of the Landau Hamiltonian in the plane with magnetic field equal to one and $\psi_\parallel$ is a fixed $L^2(\R)$-normalized function which belongs to $H^2(\R)$. With this choice we obtain
$$
\langle\psi,\Pi_X^2\psi\rangle = 2B + \ell^{-2} \|\psi_\parallel'\|^2
\qquad\text{and}\qquad
\|\Pi_X^2\psi\| \leq 2B + \ell^{-2} \|\psi_\parallel''\| \,. 
$$
If we choose $\ell$ larger than a constant times $B^{-1}$, we easily conclude that there is an $M>0$ such that for all $0\leq B\leq B_0$ and $T_1\leq T< T_c- 2T_c B\Lambda_0/\Lambda_2 - M B^2$ one has
\begin{align*}
-\Lambda_2 \frac{T_c-T}{T_c} + \Lambda_0 \langle\psi,\Pi_X^2\psi\rangle + C'(T_c-T)^2 + C''(T_c-T) \langle\psi,\Pi_X^2\psi\rangle 
+ C_T \|\Pi_X^2\psi\|^2 < 0\,.
\end{align*}
This completes the proof of part (1) of Theorem \ref{main}.\qed


\section{The approximate form of almost minimizers}

In this and the following section we work under Assumptions \ref{ass2}, \ref{ass0} and \ref{ass1}.

\subsection{The decomposition lemma}

The remainder of this paper is devoted to proving an upper bound on the critical temperature. As a preliminary step we prove in this section a decomposition lemma, which says that, if $|T_c-T|\leq C_1 B$ and if $\Phi$ satisfies $\langle\Phi,(1-V^{1/2}L_{T,B}V^{1/2})\Phi\rangle\leq C_2B$ for some fixed constants $C_1$ and $C_2$ independent of $B$, then $\Phi$ has, up to a controllable error, the same form as the trial function that we used in the lower bound on the critical temperature.

\begin{theorem}\label{decomp}
For given constants $C_1,C_2>0$ there are constants $B_0>0$ and $C>0$ such that the following holds. If $T>0$ satisfies $|T-T_c|\leq C_1B$, if $\Phi\in L^2_{\rm symm}(\R^3\times\R^3)$ satisfies $\|\Phi\|=1$ and
$$
\langle\Phi, (1- V^{1/2} L_{T,B} V^{1/2})\Phi\rangle \leq C_2 B \,,
$$
and if $\epsilon$ satisfies $\epsilon\in[B,B_0]$, then there are $\psi_\leq\in L^2(\R^3)$ and $\sigma\in L^2_{\rm symm}(\R^3\times\R^3)$ such that
$$
\Phi(x,y) = \psi_\leq((x+y)/2)\phi_*(x-y) +\sigma(x,y) \,,
$$
where
\begin{equation}
\label{eq:decomppsibounds}
\| (\Pi_X^2)^{k/2}\psi_\leq \|^2 \leq C \epsilon^{k-1} B \qquad \text{if}\ k\geq 1 \,,
\end{equation}
\begin{equation}
\label{eq:decompsigmabound}
\|\sigma\|^2 \leq C \epsilon^{-1} B
\end{equation}
and
\begin{equation}
\label{eq:decompsilower}
\|\psi_\leq\|^2 \geq 1- C\epsilon^{-1}B \,.
\end{equation}
Moreover, $\psi_\leq\in\ran\1(\Pi_X^2\leq\epsilon)$ and there is $\psi_>\in L^2(\R^3)\cap\ran\1(\Pi_X^2>\epsilon)$ such that
$$
\sigma_0(X+r/2,X-r/2):= \cos(\Pi_X\cdot r/2)\psi_>(X)\phi_*(r)
$$
satisfies
\begin{equation}
\label{eq:decompsigmabound2}
\|\sigma-\sigma_0\|^2 \leq C B_0^{-1} B
\end{equation}
and
\begin{equation}
\label{eq:decomppsibounds2}
\|\psi_>\|^2 \leq C \epsilon^{-1}B \,. 
\end{equation}
\end{theorem}

Thus, $\Phi$ is of the form $\psi_\leq(X)\phi_*(r)$ up to a small error $\sigma$. We have control on the expectation of $\Pi_X^2$ in $\psi_\leq$. However, for technical reasons we also need control on the expectation of $\Pi_X^6$. This is achieved by introducing the parameter $\epsilon$. The drawback of introducing this parameter is that the norm of error $\sigma$ deteriorates as $\epsilon$ becomes small. What will save the day is that the error $\sigma$ can be decomposed in a good part $\sigma-\sigma_0$, whose norm is controlled uniformly in $\epsilon$, and an explicit bad part $\sigma_0$, which is of a similar form as the leading term, but where the function $\psi_>$ is orthogonal to $\psi_\leq$. This will allow us to prove that the interaction between the leading term and $\sigma_0$ is of subleading order. A similar momentum cut-off for a similar purpose was already introduced in \cite{FHSS,FHSS2}.


\subsection{Upper bound on $L_{T,B}$}

Our goal in this subsection is to obtain an operator lower bound on $1-V^{1/2}L_{T,B}V^{1/2}$. In \cite{FHSS,FHSS2} such a bound was proved by means of a relative entropy inequality \cite[Lemma 3]{FHSS}, which contolled a two-particle operator by the sum of two one-particle operators, and by \cite[Lemma 5]{FHSS} which showed that the energy of the system is dominated by the kinetic energy of the center of mass motion. This was sufficient to recover the corresponding a-priori estimates. Here, we will follow a similar strategy of proof, but the argument turns out to be significantly more involved due to the fact that the components of the magnetic momentum $\Pi_X$ do not commute and because we need to keep the magnetic field in the center of mass direction.

We define the unitary operator
\begin{equation}
\label{eq:u}
U := e^{-\ii\Pi_X\cdot r/2}
\end{equation}
in $L^2(\R^3\times\R^3)$ where, as usual, $r=x-y$ and $X=(x+y)/2$.

\begin{proposition}\label{lowerboundlt}
There are constants $\delta>0$ and $C>0$ such that for all $T>0$ and $B>0$ with $B\beta(1+\beta\mu_+)\leq\delta(1+\beta\mu_-)$ one has
\begin{align*}
V^{1/2} L_{T,B} V^{1/2} & \leq \frac12 \left( U V^{1/2} \chi_\beta(p_r^2-\mu) V^{1/2} U^* + U^* V^{1/2} \chi_\beta(p_r^2-\mu) V^{1/2} U \right) \\
& \qquad + C \beta^{3/2} B \frac{( 1+ \beta\mu_+)^{3/2}}{(1+\beta\mu_-)^{3/2}} \left( \beta^{3/2} B \frac{(1+\beta\mu_+)^{3/2}}{(1+\beta\mu_-)^{3/2}} \|V\|_\infty +  \left\||\cdot| V\right\|_\infty \right).
\end{align*}
\end{proposition}

For the proof we need the following lemma which shows how the operator $U$ appears. We will use the following notation
$$
\pi_r = -\ii\nabla_r + \A(r) = -\ii\nabla_r + \frac12\bB\wedge r
$$
and
$$
\tilde\pi_r = -\ii\nabla_r + \frac12\A(r) = -\ii\nabla_r + \frac14 \bB\wedge r \,.
$$

\begin{lemma}\label{funnypi}
One has
$$
U \pi_r U^* = \tilde\pi_r + \Pi_X/2 \,,
\qquad
U^* \pi_r U = \tilde\pi_r - \Pi_X/2 \,.
$$
\end{lemma}

\begin{proof}[Proof of Lemma \ref{funnypi}]
It suffices to focus on the first two components of $\Pi_X$, which we again denote by $\Pi_X^{(1)}$ and $\Pi_X^{(2)}$, and we recall that $[\Pi_X^{(1)},\Pi_X^{(2)}]=-2\ii B$. Therefore, by the Baker--Campbell--Hausdorff formula,
$$
U = e^{-\ii(r_1\Pi_X^{(1)} + r_2 \Pi_X^{(2)})/2} = e^{-\ii r_2\Pi_X^{(2)}/2} e^{-\ii r_1\Pi_X^{(1)}/2} e^{\frac\ii 4 B r_1 r_2} \,.
$$
Thus,
\begin{align*}
\left[-\ii\partial_{r_1}, U\right] = e^{-\ii r_2\Pi_X^{(2)}/2} \left[ -\ii\partial_{r_1}, e^{-\ii r_1\Pi_X^{(1)}/2} e^{\frac\ii 4 B r_1 r_2} \right]
= U\left( -\frac{\Pi_X^{(1)}}{2} + \frac{Br_2}{4}\right).
\end{align*}
Similarly, using the Baker--Campbell--Hausdorff formula in the form
$$
U = e^{-\ii(r_1\Pi_X^{(1)} + r_2 \Pi_X^{(2)})/2} = e^{-\ii r_1\Pi_X^{(1)}/2} e^{-\ii r_2\Pi_X^{(2)}/2} e^{-\frac\ii 4 B r_1 r_2} \,.
$$
one shows that
$$
\left[-\ii\partial_{r_2}, U\right] = U\left( -\frac{\Pi_X^{(2)}}{2} - \frac{Br_1}{4}\right).
$$
Thus, we have shown that
$$
\left[p_r, U\right] = U\left( -\frac{\Pi_X}{2} - \frac{\bB\wedge r}{4}\right),
$$
which is the same as the claimed identity $\pi_r U = U(\tilde\pi_r- \Pi_X/2)$. The other identity in the lemma is proved similarly.
\end{proof}

\begin{proof}[Proof of Proposition \ref{lowerboundlt}]
Since for any real numbers $E$ and $E'$ one has
$$
\Xi_\beta(E,E') \leq \frac12 \left( \frac{\tanh\frac{\beta E}{2}}{E} +  \frac{\tanh\frac{\beta E'}{2}}{E'} \right) = \frac12 \left( \chi_\beta(E) + \chi_\beta(E') \right),
$$
we have
$$
L_{T,B} =\Xi_\beta(\ch_{B,x},\ch_{B,y}) \leq \frac12 \left(\chi_\beta(\ch_{B,x}) + \chi_\beta(\ch_{B,y}) \right).
$$
In the variables $r=x-y$, $X=(x+y)/2$ we have $\pi_x = \tilde\pi_r + \Pi_X/2$ and $\pi_y = \tilde\pi_r - \Pi_X/2$ and therefore, according to Lemma \ref{funnypi},
$$
\ch_{B,x} = (\tilde\pi_r + \Pi_X/2)^2 - \mu = U \left( \pi_r^2 -\mu\right) U^* \,,
\quad
\ch_{B,y} = (\tilde\pi_r - \Pi_X/2)^2 - \mu = U^* \left( \pi_r^2 -\mu\right) U \,.
$$
Therefore we can write the above inequality on $L_{T,B}$ as
$$
L_{T,B} \leq \frac12\left( U \chi_\beta(\pi_r^2-\mu) U^* + U^* \chi_\beta(\pi_r^2-\mu) U \right).
$$
Since $V$ commutes with $U$, we deduce that
$$
V^{1/2} L_{T,B} V^{1/2} \leq \frac12\left( U V^{1/2} \chi_\beta(\pi_r^2-\mu) V^{1/2} U^* + U^* V^{1/2} \chi_\beta(\pi_r^2-\mu) V^{1/2} U \right).
$$
In order to prove the lemma, it remains to remove the magnetic field from $\pi_r$. In terms of the Matsubara frequencies \eqref{eq:matsubara} we have
$$
\chi_\beta(\pi_r^2-\mu) = -\frac2\beta \sum_{n\in\Z} \frac{1}{\ii\omega_n (\ii\omega_n - (\pi_r^2-\mu))} \,,
$$
which follows by setting $E'=0$ in \eqref{eq:zetaidentity} and recalling $\Xi_\beta(E,0)=\chi_\beta(E)$.

For the corresponding integral kernel we obtain
$$
\chi_\beta(\pi_r^2-\mu)(r,s) = -\frac2\beta \sum_{n\in\Z} \frac{1}{\ii\omega_n} G_B^{\ii\omega_n}(r,s) = -\frac2\beta \sum_{n\in\Z} \frac{1}{\ii\omega_n} g_B^{\ii\omega_n}(r-s) e^{\frac\ii2 \bB\cdot(r\wedge s)} \,.
$$
Here we used Lemma \ref{propgbz} for the second equality. Thus, for $\phi\in L^2_{\rm symm}(\R^3)$, considered as a function of the variable $r$,
$$
\langle \phi, V^{1/2} \left( \chi_\beta(\pi_r^2-\mu) - \chi_\beta(p_r^2-\mu) \right) V^{1/2} \phi \rangle
= I_1 + I_2
$$
with
\begin{align*}
I_1 & := -\frac2\beta \sum_{n\in\Z} \frac{1}{\ii\omega_n} \iint_{\R^3\times\R^3}dr\,ds\, \overline{\phi(r)} V(r)^{1/2} \left( g_B^{\ii\omega_n} - g_0^{\ii\omega_n}\right)(r-s) e^{\frac\ii2 \bB\cdot(r\wedge s)} V(s)^{1/2} \phi(s) \,, \\
I_2 & := -\frac2\beta \sum_{n\in\Z} \frac{1}{\ii\omega_n} \iint_{\R^3\times\R^3}dr\,ds\, \overline{\phi(r)} V(r)^{1/2} g_0^{\ii\omega_n}(r-s) \left( e^{\frac\ii2 \bB\cdot(r\wedge s)}-1\right) V(s)^{1/2} \phi(s) \,.
\end{align*}
By a simple convolution inequality,
$$
|I_1| \leq \frac2\beta\sum_{n\in\Z} \frac{1}{|\omega_n|} \| V^{1/2} \phi \|^2 \|g_B^{\ii\omega_n} - g_0^{\ii\omega_n} \|_1 
\leq \frac2\beta \| V \|_\infty \| \phi \|^2 \sum_{n\in\Z} \frac{1}{|\omega_n|} \|g_B^{\ii\omega_n} - g_0^{\ii\omega_n} \|_1 
$$
and similarly, using in addition
$$
\left| e^{\frac\ii2 \bB\cdot(r\wedge s)}-1\right| = 2 \left|\sin\left( \frac14 \bB\cdot(r\wedge s) \right)\right| \leq \frac{B}{2} |r\wedge s|
$$
and
$$
|r\wedge s| = \frac12 |r\wedge (r-s)|^{1/2} |(r-s)\wedge s|^{1/2} \leq |r|^{1/2} |r-s| |s|^{1/2} \,,
$$
we get
$$
|I_2| \leq \frac2\beta \frac{B}2 \sum_{n\in\Z} \frac{1}{|\omega_n|} \| |\cdot|^{1/2} V^{1/2} \phi \|^2 \||\cdot| g_0^{\ii\omega_n} \|_1 
\leq \frac2\beta\frac{B}2 \| |\cdot| V \|_\infty \| \phi \|^2 \sum_{n\in\Z} \frac{1}{|\omega_n|} \||\cdot| g_0^{\ii\omega_n} \|_1  \,.
$$
Using the bounds from Lemma \ref{magres} we find that, if $B\beta(1+\beta\mu_+)\leq\delta(1+\beta\mu_-)$, then
$$
|I_1| \leq C \beta^3 B^2\ \frac{( 1+\beta\mu_+)^3}{(1+\beta\mu_-)^3}\ \|V\|_\infty \|\phi\|^2 \,,
$$
and using the bounds from Lemma \ref{freeres} we find that
$$
|I_2| \leq C \beta^{3/2} B \ \frac{( 1+ \beta\mu_+)^{3/2}}{(1+\beta\mu_-)^{3/2}}\ \left\||\cdot| V\right\|_\infty \|\phi\|^2 \,.
$$
(Details in getting these estimates are very similar to those explained in the proof of Proposition \ref{abstractbound} and thus not repeated here). We conclude that
\begin{align*}
& \langle \phi, V^{1/2} \left( \chi_\beta(\pi_r^2-\mu) - \chi_\beta(p_r^2-\mu) \right) V^{1/2} \phi \rangle
\leq |I_1|+ |I_2| \\
& \qquad \leq C \beta^{1/2} B \frac{( 1+ \beta\mu_+)^{3/2}}{(1+\beta\mu_-)^{1/2}} \left( \beta^{3/2} B \frac{(1+\beta\mu_+)^{3/2}}{(1+\beta\mu_-)^{3/2}} \|V\|_\infty +  \left\||\cdot| V\right\|_\infty \right) \|\phi \|^2 \,.
\end{align*}
Conjugating the resulting operator inequality by $U$ and by $U^*$ and combining it with the above inequality on $L_{T,B}$ we obtain the statement of the proposition.
\end{proof}


\subsection{A priori bound on the critical temperature and an operator inequality}\label{sec:opineq}

As a first consequence of Proposition \ref{lowerboundlt} we obtain a rough a-priori upper bound on the critical temperature.

\begin{corollary}\label{aprioriuppertemp}
There are constants $B_0>0$ and $C>0$ such that for all $0< B\leq B_0$ and $T\geq T_c+CB$ one has
$$
\langle\Phi, (1- V^{1/2} L_{T,B} V^{1/2})\Phi\rangle >0 \,,
$$
unless $\Phi=0$.
\end{corollary}

\begin{proof}
According to Proposition \ref{lowerboundlt} there are $\delta>0$ and $C>0$ such that for all $T\geq T_c$ and $0<B\beta(1+\beta\mu_+)\leq\delta(1+\beta\mu_-)$ one has the lower bound
$$
1- V^{1/2} L_{T,B} V^{1/2} \geq 1- \frac12 \left( U V^{1/2} \chi_\beta(p_r^2-\mu) V^{1/2} U^* - U^* V^{1/2} \chi_\beta(p_r^2-\mu) V^{1/2} U \right) - C B \,.
$$
(Note that the constant $C$ can be chosen independent of $T$, as long as $T\geq T_c$. In fact, the constant goes to zero as $T\to\infty$.)

We next recall that the family of operators $V^{1/2} \chi_\beta(p_r^2-\mu) V^{1/2}$ is non-decreasing with respect to $\beta$ and has an eigenvalue $1$ at $\beta=\beta_c$. Moreover, since the function $\chi_\beta(E)$ is strictly increasing with respect to $\beta$ for every $E\in\R$, we learn from analytic perturbation theory that there are $c>0$ and $T_2>T_c$ such that for all $T_c\leq T\leq T_2$,
$$
V^{1/2} \chi_\beta(p_r^2-\mu) V^{1/2} \leq 1 - c(T-T_c)\,.
$$
Again by monotonicity this implies that for all $T\geq T_c$
$$
V^{1/2} \chi_\beta(p_r^2-\mu) V^{1/2} \leq 1 - c\min\{T-T_c,T_2-T_c\}\,.
$$
Inserting this into the lower bound above we conclude that
$$
1- V^{1/2} L_{T,B} V^{1/2} \geq c\min\{T-T_c,T_2-T_c\} - C B \,.
$$
The right side is positive if $T\geq T_c + (C/c)B$ and $B\leq (c/C)(T_2-T_c)$, which proves the corollary.
\end{proof}

As a consequence of this corollary, from now on we may and will restrict ourselves to temperatures $T$ such that $|T-T_c|$ is bounded by a constant times $B$.

Our next goal is to deduce from Proposition \ref{lowerboundlt} a lower bound on the operator $1- V^{1/2} L_{T,B} V^{1/2}$. We recall that by definition of $\beta_c$ the largest eigenvalue of the operator $V^{1/2} \chi_{\beta_c}(p_r^2-\mu) V^{1/2}$ equals one. Moreover, by Assumption \ref{ass1}, this eigenvalue is simple and $\phi_*$ denotes a corresponding real-valued, normalized eigenfunction. We denote by
$$
P:=|\phi_*\rangle\langle\phi_*|
$$
the corresponding projection and $P^\bot = 1-P$. Since $V^{1/2} \chi_{\beta_c}(p_r^2-\mu) V^{1/2}$ is a compact operator, there is a $\kappa>0$ such that
\begin{equation}
\label{eq:gap}
V^{1/2} \chi_{\beta_c}(p_r^2-\mu) V^{1/2} \leq 1-\kappa P^\bot \,.
\end{equation}
Finally, we introduce the operator
\begin{equation}
\label{eq:defq}
Q:= \frac12\left( U P U^* + U^* P U \right).
\end{equation}
We can now state our operator inequality for $1- V^{1/2} L_{T,B} V^{1/2}$.

\begin{proposition}\label{opineq}
Given $C_1>0$ there are constants $B_0>0$ and $C>0$ such that for all $|T-T_c|\leq C_1B$ and $0<B\leq B_0$ one has
\begin{align}\label{eq:opineq}
1- V^{1/2} L_{T,B} V^{1/2} \geq \kappa\left( 1- Q\right) - C B \,,
\end{align}
\end{proposition}

\begin{proof}
As in the proof of Corollary \ref{aprioriuppertemp} we apply Proposition \ref{lowerboundlt} to obtain $\delta>0$ and $C>0$ such that for all $|T-T_c|\leq C_1B$ and $0<B\leq B_0 = \delta(1+\beta\mu_-)/(\beta(1+\beta\mu_+))$ one has the lower bound
$$
1- V^{1/2} L_{T,B} V^{1/2} \geq 1- \frac12 \left( U V^{1/2} \chi_\beta(p_r^2-\mu) V^{1/2} U^* + U^* V^{1/2} \chi_\beta(p_r^2-\mu) V^{1/2} U \right) - C B \,.
$$
Since the derivative of $\chi_\beta(E)$ with respect to $T$ is bounded uniformly in $E$ for $T$ close to $T_c$, we infer that there is a $C'>0$ such that for all $|T-T_c|\leq C_1 B_0$ and all $E\in\R$,
\begin{equation}
\label{eq:chiineq}
\left|\chi_\beta(E) - \chi_{\beta_c}(E)\right| \leq C'|T-T_c| \,.
\end{equation}
This, together with the gap inequality \eqref{eq:gap}, implies that for $|T-T_c|\leq C_1 B \leq C_1 B_0$,
\begin{align*}
1- V^{1/2} L_{T,B} V^{1/2} & \geq 1- \frac12 \left( U V^{1/2} \chi_{\beta_c}(p_r^2-\mu) V^{1/2} U^* + U^* V^{1/2} \chi_{\beta_c}(p_r^2-\mu) V^{1/2} U \right) \\
& \qquad\qquad -C'|T-T_c| - C B \\
& \geq \frac\kappa2 \left( U P^\bot U^* + U^* P^\bot U\right) - (C_1 C' + C) B \\
& = \kappa\left( 1- Q\right) - (C_1 C' + C) B \,,
\end{align*}
as claimed.
\end{proof}


\subsection{The operator $R$}\label{sec:rq}

We introduce the operator
\begin{equation}
\label{eq:r}
R := \int_{\R^3}dr\,  |\phi_*(r)|^2 \cos(r\cdot\Pi_X)
\end{equation}
acting in $L^2(\R^3)$. Since $-1\leq \cos(r\cdot\Pi_X)\leq 1$ and since $\phi_*$ is normalized, we have $\|R\|\leq 1$ and therefore $1-R^2\geq 0$. We now prove a more precise lower bound.

\begin{lemma}\label{rbound}
For every $B_0>0$ there are $c>0$ and $E_0>0$ such that for $0< B\leq B_0$,
$$
1-R^2\geq c \ \frac{\Pi_X^2}{E_0+\Pi_X^2} \,.
$$
\end{lemma}

For the proof of this lemma we need an auxiliary result, which is probably well-known and whose proof is included for the sake of completeness. We denote by $\Pi_X^\bot = (\Pi_X^{(1)},\Pi_X^{(2)})^T$ the first two components of $\Pi_X$. Moreover, we denote by $P^{(k)}_{2B}$ the projections on the $k$-th Landau level and by $L_k=L_k^{(0)}$ the Laguerre polynomials.

\begin{lemma}\label{diagonal}
For every $\rho>0$,
$$
\frac{1}{2\pi} \int_{\Sph^1} d\omega\, \cos(\rho \omega\cdot \Pi_X^\bot) = e^{-B\rho^2/2} \sum_{k=0}^\infty L_k(B\rho^2) P^{(k)}_{2B} \,.
$$
\end{lemma}

This lemma implies, in particular, that the operator on the left side commutes with the two-dimensional Landau Hamiltonian $(\Pi_X^\bot)^2$.

\begin{proof}[Proof of Lemma \ref{diagonal}]
Let $a=(\Pi_X^{(1)}-i\Pi_X^{(2)})/(2\sqrt B)$, so that $a^\dagger=(\Pi_X^{(1)}+i\Pi_X^{(2)})/(2\sqrt B)$ and $[a,a^\dagger]=1$. Thus, writing $\omega=(\cos\phi,\sin\phi)$ we have by the Baker--Campbell--Hausdorff formula,
$$
\exp(\ii\rho\omega\cdot\Pi_X^\bot) = \exp(\ii\rho \sqrt B (e^{-i\phi} a^\dagger + e^{i\phi} a)) = e^{-B\rho^2/2} \exp(\ii\rho \sqrt B e^{-i\phi} a^\dagger ) \exp(\ii\rho \sqrt B e^{i\phi} a).
$$
Expanding the exponential, we find
\begin{align*}
\frac1{2\pi} \int_{\Sph^1} d\omega\, \cos(\rho \omega\cdot \Pi_X^\bot) & = \frac1{2\pi}\int_{\Sph^1}d\omega\, \exp(\ii\rho \omega\cdot \Pi_X^\bot) \\
& =e^{-B\rho^2/2} \sum_{n,m=0}^\infty \frac{1}{n!\, m!} \frac1{2\pi} \int_{-\pi}^\pi d\phi\, \left( \ii\rho\sqrt B e^{-i\phi} a^\dagger\right)^n \left( \ii\rho\sqrt B e^{i\phi} a \right)^m \\
& = e^{-B\rho^2/2} \sum_{n=0}^\infty \frac{(-\rho^2 B)^n}{(n!)^2} (a^\dagger)^n a^n \,.
\end{align*}
It is well-known that there is a basis of the $k$-th Landau level of the form $(a^\dagger)^k (b^\dagger)^\ell |0\rangle$, $\ell\in\N_0$, where $b$ and $b^\dagger$ correspond to an independent oscillator. Since
$$
(a^\dagger)^n a^n (a^\dagger)^k (b^\dagger)^\ell |0\rangle =
\begin{cases}
k(k-1)\ldots(k-n+1) (a^\dagger)^k (b^\dagger)^\ell |0\rangle & \text{if}\ k\geq n \,,\\
0 & \text{if}\ k<n \,,
\end{cases}
$$
we deduce that
\begin{align*}
& \frac1{2\pi} \int_{\Sph^1} d\omega\, \cos(\rho \omega\cdot \Pi_X^\bot) (a^\dagger)^k (b^\dagger)^\ell |0\rangle \\
&\qquad = e^{-B\rho^2/2} \sum_{n=0}^k \frac{(-\rho^2 B)^n}{(n!)^2} k(k-1)\cdots(k-n+1) (a^\dagger)^k (b^\dagger)^\ell |0\rangle \\
& \qquad = e^{-B\rho^2/2} L_k(\rho^2 B) (a^\dagger)^k (b^\dagger)^\ell |0\rangle \,,
\end{align*}
where we used \cite[(22.3.9)]{AbSt}. This proves the claimed formula.
\end{proof}

\begin{proof}[Proof of Lemma \ref{rbound}]
The operator $R$ can be diagonalized explicitly by performing a Fourier transform in the $x_3$ variable and by decomposing into Landau levels in the $(x_1,x_2)$ variables. The operator then acts as multiplication by the numbers
\begin{equation} 
\label{rnP}
R_{k,p_3} = \pi \int_0^\infty du \int_{\mathbb{R}}dx_3\, |\varphi_*(\sqrt{u+x_3^2})|^2
\cos(x_3 p_3 ) e^{-(1/2)Bu} L_k(Bu) \,,
\quad k\in\N_0 \,,\ p_3\in\R \,.
\end{equation}
To obtain this formula we note that, since $\phi_*$ is spherically symmetric,
\begin{align*}
R & = \int_{\R^3} dr\, |\phi_*(r)|^2 e^{\ii r\cdot\Pi_X} = \int_0^\infty d\rho\,\rho \int_\R dx_3 \, |\phi_*(\sqrt{\rho^2+x_3^2})|^2 e^{\ii x_3 p_3} \int_{\Sph^1} d\omega\, e^{\ii\rho\omega\cdot\Pi_X^\bot} \\
& = \int_0^\infty d\rho\,\rho \int_\R dx_3 \, |\phi_*(\sqrt{\rho^2+x_3^2})|^2 \cos(x_3 p_3) \int_{\Sph^1} d\omega\, \cos(\rho\omega\cdot\Pi_X^\bot) \,.
\end{align*}
We now apply Lemma \ref{diagonal} and obtain
\begin{align*}
R & = 2\pi \sum_{k=0}^\infty P_{2B}^{(k)} \int_0^\infty d\rho\,\rho \int_{\R^3} dx_3\, |\phi_*(\sqrt{\rho^2+x_3^2})|^2 \cos(x_3 p_3) e^{-B\rho^2/2} L_k(B\rho^2) \,.
\end{align*}
Changing variables $u=\rho^2$ and performing a Fourier transform in the $x_3$ variable we obtain \eqref{rnP}.

Since in the same representation $\Pi_X^2$ becomes multiplication by
\begin{equation} 
\label{EnP}
E_{k,p_3}=2B(2k+1) + p_3^2 \,, \quad k\in\N_0\,,\ p_3\in\mathbb{R} \,,
\end{equation} 
we need to prove that for $0< B\leq B_0$,
$$
R_{k,p_3}^2 \leq 1 - c \frac{E_{k,p_3}}{E_0 + E_{k,p_3}} \,, \quad k\in\N_0\,,\ p_3\in\mathbb{R} \,.
$$
We prove this inequality separately for small, medium and large values of $E_{k,p_3}$.

\emph{Step 1.} We show that there are constants $E_*>0$ and $C>0$ such that for all $B>0$, $k\in\N_0$ and $p_3\in\R$ with $E_{k,p_3}\leq E_*$ we have
\begin{equation}\label{eq:rbound1}
|R_{k,p_3}| \leq 1- C E_{k,p_3} \,. 
\end{equation} 

Indeed, using the inequalities \eqref{eq:cosineq} we get from \eqref{eq:r} the operator inequalities
\begin{align*}
0 \leq R - \int_{\mathbb{R}^3}dr\, |\varphi_*(r)|^2\left(1-\frac12 (r\cdot\Pi_X)^2 \right) \leq  \frac1{24} \int_{\mathbb{R}^3}dr\, |\varphi_*(r)|^2 (r\cdot\Pi_X)^4 \,.
\end{align*}
Let us abbreviate
\begin{equation} 
\langle |r|^{2m}\rangle :=   \int_{\mathbb{R}^3} dr\, |\varphi_*(r)|^2|r|^{2m}\,, \quad m=1,2 \,, 
\end{equation} 
and note that these numbers are finite by the decay properties of $\phi_*$ \cite[Proposition~1]{FHSS}. We now compute, as in \eqref{eq:pisquared},
$$
\int_{\R^3} dr\, |\phi_*(r)|^2 (r\cdot\Pi_X)^2 = \frac13 \langle |r|^2\rangle \Pi_X^2
$$
and bound, using \eqref{eq:threeoperators},
$$
\int_{\R^3} dr\, |\phi_*(r)|^2 (r\cdot\Pi_X)^4 \leq C \langle |r|^4\rangle (\Pi_X^2)^2 \,.
$$
We obtain
$$
1 -\frac16\langle|r|^2\rangle \Pi_X^2\leq R \leq 1 -\frac16\langle|r|^2\rangle \Pi_X^2 + C \langle|r|^4\rangle (\Pi_X^2)^2 \,, 
$$
or, equivalently,
\begin{equation*} 
1 -\frac16\langle|r|^2\rangle E_{k,p_3} \leq R_{k,p_3} \leq 1 -\frac16\langle|r|^2\rangle E_{k,p_3} + C \langle|r|^4\rangle E_{k,p_3}^2 
\quad\text{for all}\ k\in\N_0\,,\ p_3\in\R \,. 
\end{equation*}
This implies the claimed bound \eqref{eq:rbound1}.

\emph{Step 2.} We show that
\begin{equation}
\label{eq:rbound2}
R_{k,p_z} \to 0 \ \text{as}\ E_{k,p_z} \to \infty\ \text{uniformly for all sufficiently small}\ B \,,
\end{equation}
that is, for every $B_0>0$ and $\epsilon>0$ there is an $E_*>0$ such that for all $k\in\N_0$ and $p_3\in\R$ with $E_{k,p_3}\geq E_*$ one has $|R_{k,p_3}|\leq\epsilon$.

This follows by a Riemann--Lebesgue-type argument. Indeed, for given $\epsilon>0$ we choose $f\in C_c^\infty(0,\infty)$ such that
$$ 
\pi \int_0^\infty du \int_{\mathbb{R}}dz\,\left|  |\varphi_*(\sqrt{u+z^2})|^2 - f(\sqrt{u+z^2})\right| \leq \varepsilon/2.
$$
Let
$$
\tilde R_{k,p_3} = \pi \int_0^\infty du \int_{\mathbb{R}}dx_3\, f(\sqrt{u+x_3^2})
\cos(x_3 p_3 ) e^{-(1/2)Bu} L_k(Bu) \,.
$$
Since (see \cite[(22.14.12)]{AbSt})
\begin{equation}
\label{eq:boundedbyone}
-1\leq \cos(x_3 p_3 ) \leq 1,\quad -1\leq e^{-(1/2)Bu} L_k(Bu)\leq 1 \,,
\end{equation}
we find that
$$
\left| R_{k,p_3} \right| \leq \left| \tilde R_{k,p_3} \right| + \epsilon/2 \,,
$$
and we are reduced to proving that $\tilde R_{k,p_3}\to 0$ as $E_{k,p_3}\to\infty$.

We prove two different bounds on $\tilde R_{k,p_3}$. First, we write
$$
\frac{d}{dx} \sin x = \cos x 
$$
and obtain
\begin{align*}
\tilde R_{k,p_3} & = \frac{\pi}{p_3} \int_0^\infty du \int_{\mathbb{R}}dx_3\, f(\sqrt{u+x_3^2}) \frac{d}{dx_3} \sin(x_3 p_3) e^{-(1/2)Bu} L_k(Bu) \\
& = - \frac{\pi}{p_3} \int_0^\infty du \int_{\mathbb{R}}dx_3\, f'(\sqrt{u+x_3^2}) \frac{x_3}{\sqrt{u+x_3^2}} \sin(x_3 p_3) e^{-(1/2)Bu} L_k(Bu) \,.
\end{align*}
Therefore, using \eqref{eq:boundedbyone},
$$
\left| \tilde R_{k,p_3} \right| \leq \frac{\pi}{|p_3|} \int_0^\infty du \int_{\mathbb{R}}dx_3\, \left| f'(\sqrt{u+x_3^2}) \right| \,.
$$
This is the first bound. For the second bound, we write
$$ 
\frac d{dx} (x L_k^{(1)}(x) ) = (k+1) L_k(x)
$$
with the generalized Laguerre polynomials $L_k^{(1)}$ and obtain
\begin{align*}
\tilde R_{k,p_3} & = \frac{\pi}{B(k+1)} \int_0^\infty du \int_{\mathbb{R}}dx_3\, f(\sqrt{u+x_3^2}) \cos(x_3 p_3) e^{-(1/2)Bu} \frac d{du} (Bu L_k^{(1)}(Bu)) \\
& = - \frac{\pi}{2B(k+1)} \int_0^\infty du \int_{\mathbb{R}}dx_3\left( f'(\sqrt{u+x_3^2}) \frac{1}{\sqrt{u+x_3^2}} - B f(\sqrt{u+x_3^2}) \right) \\
& \qquad \qquad\qquad\qquad \times \cos(x_3 p_3) e^{-(1/2)Bu} Bu L_k^{(1)}(Bu) \,.
\end{align*}
We now use the fact that for every $M>0$ there is a $C>0$ such that
$$
x \left| L_k^{(1)}(x) \right| e^{-x/2} \leq C ((k+1)x)^{1/4}
\qquad\text{for all}\ x\in [0,M] \,.
$$
(This bound is a consequence of the more precise uniform asymptotics in \cite{Er} for $k\geq k_0$. In fact, the bound is valid for $x\in[0,4b(k+1)]$ with any fixed $b<1$, and it can be further improved for $x\leq C (k+1)^{-1}$, but the stated bound suffices for our purposes. The bound for $k<k_0$ is immediate.) Using this bound with $x=Bu$ (which is bounded from above since $f$ has compact support and $B\leq B_0$) we can bound
\begin{align*}
\left| \tilde R_{k,p_3} \right|\leq  \frac{\pi C}{(2B(k+1))^{3/4}} \int_0^\infty du \int_{\mathbb{R}}dx_3\left( \frac{|f'(\sqrt{u+x_3^2})|}{\sqrt{u+x_3^2}} + B_0 |f(\sqrt{u+x_3^2})| \right) u^{1/4}  \,.
\end{align*}

Combining the two bounds we see that
$$
\left| \tilde R_{k,p_3} \right|\leq C' \min\{ |p_3|^{-1}, (B(k+1))^{-3/4}\} \,,
$$
and this is bounded by $C'' \max\{ E_{k,p_3}^{-1/2}, E_{k,p_3}^{-3/4} \}$, which is $\leq \epsilon/2$ if $E_{k,p_3}$ is large enough. This proves \eqref{eq:rbound2}.

\emph{Step 3.} We show that for any $E_\leq<E_\geq$ and any $B_0>0$ there is a $c>0$ such that for all $0<B\leq B_0$, $k\in\N_0$ and $p_3\in\R$ with $E_\leq \leq E_{k,p_3}\leq E_\geq$ one has
\begin{equation} 
\label{eq:rbound3}
\left| R_{k,p_3}\right |\leq 1-c \,.
\end{equation} 

From inequalities \eqref{eq:boundedbyone} and the fact that the equality $\cos(p_3 x_3) e^{-(1/2)Bu}L_n(Bu)=\pm 1$ holds only on a set of $(u,x_3)$ of measure zero we conclude that $| R_{k,p_3} |<1$ for all $k\in\N_0$ and $p_3\in\R$. Therefore, arguing by contradiction, if \eqref{eq:rbound3} were wrong, there would be a sequence $(B_j,k_j,{p_3}_j)$ with $0<B_j\leq B_0$, $E_\leq\leq E^{(j)}_{k_j,{p_3}_j}\leq E_\geq$ such that
$$
\left| R^{(j)}_{k_j,{p_3}_j} \right| \to 1 \,.
$$
(With the superscript on $E_{k,p_3}^{(j)}$ and $R_{k,p_3}^{(j)}$ we indicate that the corresponding quantities are evaluated at $B=B_j$.) After passing to a subsequence we may assume that ${p_3}_j\to p_3$, $B_j \to B$ and $E_{k_j,{p_3}_j}^{(j)}\to E$ for some $p_3\in\R$, $0\leq B\leq B_0$ and $E_\leq \leq E\leq E_\geq$. If $B>0$, then we may assume that $k_j\to k$ for some $k\in\N_0$ and we easily deduce that $|R_{k,p_3}|=1$, which is a contradiction. It remains to discuss the case $B=0$. Writing $B_j = (E_{k_j,{p_3}_j}^{(j)}-{p_3}_j^2)/(2(2k_j+1))\sim (E - p_3^2)/(4k_j)$ and using the fact that
$$ 
e^{-\frac{x}{2k}} L_k( x/k) \to J_0(2\sqrt{x})
\qquad\text{as}\ k\to\infty
$$
for all $x \geq 0$ (see \cite[(22.15.2)]{AbSt}), where $J_0$ is the Bessel function of the first kind of order zero, we deduce the dominated convergence that
$$
R^{(j)}_{n_j,{p_3}_j} \to \pi \int_0^\infty du \int_{\mathbb{R}}dx_3\, |\varphi_*(\sqrt{u+x_3^2})|^2
\cos(x_3 p_3 ) J_0(\sqrt{(E-p_3^2) u}) \,.
$$
Thus,
$$
\left| \pi \int_0^\infty du \int_{\mathbb{R}}dx_3\, |\varphi_*(\sqrt{u+x_3^2})|^2
\cos(x_3 p_3 ) J_0(\sqrt{(E-p_3^2)u}) \right| = 1 \,.
$$
Since $|\cos (a) J_0(b)|\leq 1$ for all $a,b$ (see \cite[(9.1.60)]{AbSt}) with strict inequality away from a set of measure zero, we obtain again a contradiction. This proves \eqref{eq:rbound3} and therefore concludes the proof of the lemma.
\end{proof}


\subsection{Proof of the decomposition lemma}

As a consequence of Proposition~\ref{opineq} we now deduce a first decomposition result for almost maximizers $\Phi$ of $1-V^{1/2} L_{T,B} V^{1/2}$. 

\begin{lemma}\label{decomp1}
Given $C_1,C_2>0$ there are $B_0>0$, $E_0>0$ and $C>0$ with the following properties. If $|T-T_c|\leq C_1 B\leq C_1 B_0$ and if $\Phi\in L^2_{\rm symm}(\R^3\times\R^3)$ with $\|\Phi\|=1$ satisfies
\begin{equation}
\label{eq:almostmin}
\langle\Phi, (1- V^{1/2} L_{T,B} V^{1/2})\Phi\rangle \leq C_2 B \,,
\end{equation}
then there are $\psi\in L^2(\R^3)$ and $\xi\in L^2_{\rm symm}(\R^3\times\R^3)$ such that
$$
\Phi(X+\frac r2,X-\frac r2) = \cos(\Pi_X\cdot r/2)\psi(X) \phi_*(r) + \xi(X+\frac r2,X-\frac r2)
$$
with
$$
\left\langle\psi,\frac{\Pi_X^2}{E_0+\Pi_X^2}\psi\right\rangle + \|\xi\|^2 \leq C B \,.
$$
and
$$
\|\psi\|^2 \geq 1- CB \,.
$$
\end{lemma}

\begin{proof}
We begin the proof with some preliminary remarks. Let us introduce the operator $A: L^2_{\rm symm}(\R^3\times\R^3)\to L^2(\R^3)$ by
$$
(A\Phi)(X) := \int_{\R^3} dr\, \overline{\phi_*(r)} \cos(\Pi_X\cdot r/2) \Phi(X+r/2,X-r/2) \,.
$$
A simple computation shows that its adjoint $A^*:L^2(\R^3)\to L^2_{\rm symm}(\R^3\times\R^3)$ satisfies
$$
(A^*\psi)(X+r/2,X-r/2) = \cos(\Pi_X\cdot r/2) \psi(X)\phi_*(r) \,.
$$
Note that this is the form of the leading term in the decomposition of $\Phi$. Recalling definition \eqref{eq:defq} of $Q$, as well as the fact that $Q$ and $A$ act on symmetric functions, we find
\begin{equation}
\label{eq:astara}
A^* A = Q \,.
\end{equation}
On the other hand, using $(1/2)(1+\cos\lambda) = \cos^2(\lambda/2)$ we see that
\begin{equation}
\label{eq:aastar}
A A^* = \frac12 (1+R)
\end{equation}
with $R$ from \eqref{eq:r}. Next, we observe that the operator $1+R$ is boundedly invertible. In fact, since $1-R\leq 2$, we have from Lemma \ref{rbound}
$$
1+R \geq \frac12 (1-R^2) \geq \frac c2 \ \frac{\Pi_X^2}{E_0+\Pi_X^2} \geq \frac c2 \ \frac{2B}{E_0+2B} >0 \,.
$$

Now let $\Phi$ be as in the statement of the lemma. By Proposition \ref{opineq} and assumption \eqref{eq:almostmin} we obtain that
\begin{equation}
\label{eq:apriorisize}
\langle\Phi,(1-Q)\Phi\rangle \leq \kappa^{-1} (C+C_2)B \,.
\end{equation}
If we define
$$
\psi := \frac{2}{1+R} A\Phi \,,
$$
then the decomposition of $\Phi$ holds with
$$
\xi := \Phi - A^* \psi = \Phi - A^* \frac{2}{1+R} A\Phi \,.
$$
Because of \eqref{eq:aastar} we have
$$
A \left( 1-A^* \frac{2}{1+R} A \right) = \left( 1 - AA^* \frac{2}{1+R} \right) A = 0 \,,
$$
and therefore $\langle A^*\psi,\xi\rangle=0$. This implies
\begin{equation}
\label{eq:psinorm}
1=\|\Phi\|^2 = \|A^*\psi\|^2 + \|\xi\|^2 \leq \|\psi(X)\phi_*(r)\|^2 + \|\xi\|^2 = \|\psi\|^2 + \|\xi\|^2 \,.
\end{equation}
Because of \eqref{eq:astara} and \eqref{eq:aastar} we have
\begin{align*}
A \left( 1- Q \right) \left( 1-A^* \frac{2}{1+R} A \right) & = A \left( 1-A^*A \right) \left( 1-A^* \frac{2}{1+R} A \right) \\
& = \left( 1- A A^* \right) \left( 1 - AA^* \frac{2}{1+R} \right) A = 0 \,,
\end{align*}
and
$$
Q \left( 1- A^* \frac{2}{1+R}A \right) = A^* \left( 1 - AA^* \frac{2}{1+R} \right)A = 0 \,,
$$
and therefore $\langle A^*\psi,(1-Q)\xi\rangle=0$ and $\langle\xi,Q\xi\rangle =0$. This implies
\begin{equation}
\label{eq:apriorienergydecomp}
\left\langle\Phi, \left( 1- Q\right)\Phi\right\rangle
= \left\langle A^*\psi, \left( 1- Q\right) A^*\psi\right\rangle + \|\xi\|^2 \,.
\end{equation}
Finally, we compute, using \eqref{eq:aastar},
$$
A \left(1-Q\right) A^* = A \left( 1-A^* A \right) A^* = \left( 1- AA^*\right) A^*A = \frac{1-R}{2} \frac{1+R}{2} = \frac{1-R^2}{4}
$$
and obtain
\begin{equation}
\label{eq:aprioric}
\left\langle A^*\psi, \left( 1- Q\right)A^*\psi\right\rangle = \frac14 \langle\psi,\left(1-R^2\right)\psi\rangle \,.
\end{equation}
The claimed bounds on $\langle\psi,\Pi_X^2(E_0+\Pi_X^2)^{-1}\psi\rangle$ and $\|\xi\|^2$ follow from \eqref{eq:apriorisize}, \eqref{eq:apriorienergydecomp} and \eqref{eq:aprioric} together with Lemma \ref{rbound}. Inserting the bound on $\|\xi\|^2$ into \eqref{eq:psinorm} we obtain the claimed bound on $\|\psi\|^2$.
\end{proof}

Finally, we can provide a proof of the decomposition lemma.

\begin{proof}[Proof of Theorem \ref{decomp}]
Let $\psi$ be as in Lemma \ref{decomp1}. For $\epsilon\in[B,B_0]$ we set
$$
\psi_\leq := \1(\Pi_X^2 \leq \epsilon)\psi \,,
\qquad
\psi_> := \1(\Pi_X^2 > \epsilon)\psi \,.
$$
Recall from Lemma \ref{decomp1} that $\langle\psi,\Pi_X^2(E_0+\Pi_X^2)^{-1}\psi\rangle\leq CB$. This implies that for $k\geq 1$,
\begin{align}
\label{eq:aprioripi}
\left\|\left(\Pi_X^2\right)^{k/2} \psi_\leq\right\|^2 & \leq \epsilon^{k-1} \|\Pi_X\psi_\leq\|^2 \leq (E_0+\epsilon) \epsilon^{k-1} \left\langle\psi,\frac{\Pi_X^2}{E_0+\Pi_X^2}\psi\right\rangle \notag \\
& \leq C (E_0+\epsilon) \epsilon^{k-1} B
\end{align}
and
\begin{equation}
\label{eq:apriorioutside}
\|\psi_>\|^2\leq \frac{E_0+\epsilon}{\epsilon}\left\langle\psi,\frac{\Pi_X^2}{E_0+\Pi_X^2}\psi\right\rangle \leq C \frac{E_0+\epsilon}{\epsilon} B \,.
\end{equation}

We now define
$$
\sigma_0 := \cos(\Pi_X\cdot r/2) \psi_>(X) \phi_*(r)
$$
and
$$
\sigma:= \sigma_0 + \xi + \left( \cos(\Pi_X\cdot r/2) - 1 \right) \psi_\leq(X) \phi_*(r) \,,
$$
so that, by Lemma \ref{decomp1}, $\Phi= \psi_\leq(X)\phi_*(r) + \sigma$. According to Lemma \ref{decomp1} and \eqref{eq:apriorioutside}, we have
$$
\|\psi_\leq\|^2 = \|\psi\|^2 - \|\psi_>\|^2 \geq 1- CB - C \epsilon^{-1}B\geq 1- C'\epsilon^{-1} B \,.
$$

We now prove the claimed bounds on $\sigma_0$ and $\sigma-\sigma_0$. According to \eqref{eq:apriorioutside}, we have
$$
\|\sigma_0\|^2 \leq C \frac{E_0+\epsilon}{\epsilon} B \,.
$$
Moreover, for each $r\in\R^3$, since $1-\cos\lambda=2\sin^2(\lambda/2)\leq\lambda$,
$$
\left\| \left( \cos(\Pi_X\cdot r/2)-1 \right) \psi_\leq \right\|^2 \leq \frac12 \|r\cdot\Pi_X \psi_\leq\|^2 \leq \frac{r^2}{2} \|\Pi_X \psi_\leq\|^2 \leq \frac{C}{2} (E_0+\epsilon)B r^2 \,,
$$
where we used \eqref{eq:aprioripi}. Thus,
$$
\left\| \left(\cos(\Pi_X\cdot r/2)-1\right) \psi_\leq \phi_* \right\|^2 \leq \frac{C}{2} (E_0+\epsilon)B \int_{\R^3} dr\, |\phi_*(r)|^2r^2 \,.
$$
The last integral is finite by the decay properties of $\phi_*$, see \cite[Proposition 1]{FHSS}. Recalling the bound on the norm of $\xi$ from Lemma \ref{decomp1} we finally obtain the claimed bound $\|\sigma-\sigma_0\|^2 \leq C' B$. (Note that the bound on $\left(\cos(\Pi_X\cdot r/2)-1\right) \psi_\leq \phi_*$ could be improved to $B\epsilon$ if we use $1-\cos\lambda\leq\lambda^2/2$, but this does not improve the final bound on $\sigma-\sigma_0$.) This proves the theorem.
\end{proof}



\section{Upper bound on the critical temperature}

In this section we prove part (2) of Theorem \ref{main}. In view of Corollary \ref{aprioriuppertemp} it suffices to consider $T$ satisfying $|T-T_c|\leq C_1B$. Moreover, it clearly suffices to consider functions $\Phi$ with $\|\Phi\|=1$ satisfying
$$
\langle\Phi, (1- V^{1/2} L_{T,B} V^{1/2})\Phi\rangle \leq C_2 B
$$
(for if there are no such $\Phi$, then the theorem is trivially true). According to Theorem~\ref{decomp}, for any parameter $\epsilon\in [B,B_0]$, $\Phi$ can be decomposed as
$$
\Phi = \psi_\leq(X)\phi_*(r) + \sigma \,.
$$
Thus,
$$
\langle\Phi, (1- V^{1/2} L_{T,B} V^{1/2})\Phi\rangle = I_1 + I_2 + I_3
$$
with
\begin{align*}
I_1 & = \langle \psi_\leq(X)\phi_*(r), (1- V^{1/2} L_{T,B} V^{1/2}) \psi_\leq(X)\phi_*(r) \rangle \,,\\
I_2 & = \langle\sigma, (1- V^{1/2} L_{T,B} V^{1/2})\sigma\rangle \,,\\
I_3 & = 2\re \langle\sigma, (1- V^{1/2} L_{T,B} V^{1/2})\psi_\leq(X)\phi_*(r) \rangle \,.
\end{align*}

The term $I_1$ is the main term and we have, exactly as in the proof of the lower bound on the critical termperature,
$$
I_1 \geq -\Lambda_2 \frac{T_c-T}{T_c}\|\psi_\leq\|^2 + \Lambda_0 \langle\psi_\leq,\Pi_X^2\psi_\leq\rangle - C \epsilon B \,.
$$
Here the remainder $\epsilon B$ comes from the bound on $\|\Pi_X^2\psi_\leq\|^2$ in \eqref{eq:decomppsibounds}.

Let us bound the term $I_2$. Using the operator inequality from Proposition \ref{opineq} (dropping the non-negative term $\kappa(1-Q)$) and recalling that $T\geq T_c - C_1 B$ we obtain
$$
I_2 \geq - C B \|\sigma\|^2 \geq - C' \epsilon^{-1} B^2 \,,
$$
where we used the bound on $\sigma$ from \eqref{eq:decompsigmabound}.

It remains to bound $I_3$. According to Lemmas \ref{lltilde} and \ref{ltildem}, we have
\begin{align*}
I_3 & \geq \langle\sigma, (1- V^{1/2} N_{T,B} V^{1/2})\psi_\leq(X)\phi_*(r)\rangle - C B \|\sigma\| \|\psi_\leq\| \\
& \geq \langle\sigma, (1- V^{1/2} N_{T,B} V^{1/2})\psi_\leq(X)\phi_*(r)\rangle - C B^{3/2} \epsilon^{-1/2} \,.
\end{align*}
Here in the first inequality we used the assumption that $(1+|r|)V$ is bounded and in the second inequality we used \eqref{eq:decomppsibounds} and \eqref{eq:decompsigmabound}.

We next decompose
\begin{align*}
& \langle\sigma, (1- V^{1/2} N_{T,B} V^{1/2})\psi_\leq(X)\phi_*(r)\rangle  = \langle\sigma, (1- V^{1/2} \chi_\beta(p_r^2-\mu) V^{1/2})\psi_\leq(X)\phi_*(r)\rangle \\
& \qquad\qquad\qquad\qquad\qquad\qquad\qquad\quad + \langle\sigma, V^{1/2}(\chi_\beta(p_r^2-\mu)- N_{T,B}) V^{1/2}\psi_\leq(X)\phi_*(r)\rangle \,.
\end{align*}
We have, in view of \eqref{eq:chiineq}, $|T-T_c|\leq C_1B$ and the boundedness of $V$,
\begin{align*}
& \langle\sigma, (1- V^{1/2} \chi_\beta(p_r^2-\mu) V^{1/2})\psi_\leq(X)\phi_*(r)\rangle \\
& \qquad \geq \langle\sigma, (1- V^{1/2} \chi_{\beta_c}(p_r^2-\mu) V^{1/2})\psi_\leq(X)\phi_*(r)\rangle - C B \|\sigma\| \|\psi_\leq\| \\
& \qquad\geq - C' B^{3/2} \epsilon^{-1/2} \,,
\end{align*}
where we again used \eqref{eq:decomppsibounds} and \eqref{eq:decompsigmabound}. Thus, we are left with the term involving the difference $2\chi_\beta(p_r^2-\mu)- N_{T,B}$. According to Lemma \ref{mchi} we have
$$
\langle\sigma-\sigma_0, V^{1/2}(\chi_\beta(p_r^2-\mu)- N_{T,B}) V^{1/2}\psi_\leq(X)\phi_*(r)\rangle \geq - C \|\sigma-\sigma_0\| \|\Pi_X^2\psi_\leq\| \geq - C' B \epsilon^{1/2} \,,
$$
where we used \eqref{eq:decompsigmabound2}. Finally, using the explicit form of $\sigma_0$ we write
\begin{align*}
& \langle\sigma_0, V^{1/2}(\chi_\beta(p_r^2-\mu)- N_{T,B}) V^{1/2}\psi_\leq(X)\phi_*(r)\rangle \\
& \quad = \langle \psi_>(X)\phi_*(r), V^{1/2}(\chi_\beta(p_r^2-\mu)- N_{T,B}) V^{1/2}\psi_\leq(X)\phi_*(r)\rangle \\
& \quad\quad + \langle (\cos(\Pi_X\cdot r/2)-1) (\psi_>(X) \phi_*(r)), V^{1/2}(\chi_\beta(p_r^2-\mu)- N_{T,B}) V^{1/2}\psi_\leq(X)\phi_*(r)\rangle \,.
\end{align*}
Since the operator $V^{1/2}(2\chi_\beta(p_r^2-\mu)- N_{T,B}) V^{1/2}$ commutes with $\Pi_X^2$ and since $\psi_>$ and $\psi_\leq$ are localized where $\Pi_X^2>\epsilon$ and where $\Pi_X^2\leq\epsilon$, respectively, we have
$$
\langle \psi_>(X)\phi_*(r), V^{1/2}(\chi_\beta(p_r^2-\mu)- N_{T,B}) V^{1/2}\psi_\leq(X)\phi_*(r)\rangle = 0 \,.
$$
In order to bound the other term we decompose $\cos(\Pi_X\cdot r/2)-1=(1/2)(U-1)+(1/2)(U^*-1)$. Since $U$ commutes with $V$ we have, by Lemma \ref{mchicomm},
\begin{align*}
& \left|\langle (U-1) (\psi_>(X) \phi_*(r)), V^{1/2}(\chi_\beta(p_r^2-\mu)- N_{T,B}) V^{1/2}\psi_\leq(X)\phi_*(r)\rangle \right| \\
& \quad = \left|\langle \psi_>(X) \phi_*(r), V^{1/2}(U^*-1)(\chi_\beta(p_r^2-\mu)- N_{T,B}) V^{1/2}\psi_\leq(X)\phi_*(r)\rangle \right| \\
& \quad\leq \|\psi_>\| \||r|V^{1/2}\phi_*\| \left\||r|^{-1} (U^*-1)(\chi_\beta(p_r^2-\mu)- N_{T,B}) V^{1/2}\psi_\leq(X)\phi_*(r) \right\| \\
& \quad \leq C\|\psi_>\| \|\Pi_X^3\psi_\leq\| \\
& \quad \leq C' B \epsilon^{1/2} \,.
\end{align*}
Here we used \eqref{eq:decomppsibounds} and \eqref{eq:decomppsibounds2}. Note that the factor of $|r|$ presents no problem because $|r|V\in L^\infty$ and $|r|^{1/2}\phi_*\in L^2$ by \cite[Prop. 1]{FHSS}. (We note that this bound could be improved by not splitting $\cos(\Pi_X\cdot r/2)-1$ into two terms and by proving a version of Lemma \ref{mchicomm} with $\cos(\Pi_X\cdot r/2)-1$ instead of $U-1$. This would lead to a bound of the form $B\epsilon$, which, however, would not improve the final result because the bound on $\sigma-\sigma_0$ is also of the form $B\epsilon^{1/2}$.)

To summarize, we have shown that
\begin{align*}
\langle\Phi, (1- V^{1/2} L_{T,B} V^{1/2})\Phi\rangle & \geq -\Lambda_2 \frac{T_c-T}{T_c}\|\psi_\leq\|^2 + \Lambda_0 \langle\psi_\leq,\Pi_X^2\psi_\leq\rangle \\
& \qquad - C B \left( \epsilon + \epsilon^{-1} B + B^{1/2} \epsilon^{-1/2} + \epsilon^{1/2} \right).
\end{align*}
In order to minimize the error, we choose $\epsilon= B^{1/2}$ and obtain
$$
\langle\Phi, (1- V^{1/2} L_{T,B} V^{1/2})\Phi\rangle \geq -\Lambda_2 \frac{T_c-T}{T_c}\|\psi_\leq\|^2 + \Lambda_0 \langle\psi_\leq,\Pi_X^2\psi_\leq\rangle - C B^{5/4} \,.
$$
We estimate $\langle\psi_\leq,\Pi_X^2\psi_\leq\rangle \geq 2B \|\psi_\leq\|^2$ and obtain
$$
\langle\Phi, (1- V^{1/2} L_{T,B} V^{1/2})\Phi\rangle \geq B \left( \left( 2 \Lambda_0-\Lambda_2 \frac{T_c-T}{T_c B} \right) \|\psi_\leq\|^2 - C B^{1/4} \right) \,.
$$
Recalling that $\|\psi_\leq\|^2 \geq c>0$ we finally conclude that
$$
\langle\Phi, (1- V^{1/2} L_{T,B} V^{1/2})\Phi\rangle >0
$$
provided that $T> T_c - 2T_c(\Lambda_0/\Lambda_2)B + C B^{5/4}$. This concludes the proof of the upper bound on the critical temperature.


\appendix

\section{Comparison with WHH theory}\label{appendix}

In this appendix we show that our main result is consistent with what is known from WHH theory in the physics literature.
   
The model studied in this paper describes a system of electrons characterized by the dispersion relation $p^2-\mu$ (choosing units $\hbar=k_B=e=2m^*=1$) and an attractive two body interaction $-2V(x-y)$ depending only on the distance $|x-y|$, and the equations we use are the same as the ones underlying WHH theory \cite{HeWe,WeHeHo}. WHH theory is based on two approximations: (i) the {\em local approximation}, which (essentially) amounts to replacing $V(x-y)$ by a local interaction $g\delta(x-y)$, (ii) the {\em phase approximation}, which in our notation amounts to replacing $L_{T,B}$ by the operator $N_{T,B}$ in \eqref{eq:defmt}. Moreover, in WHH theory the function $\alpha(x,y)$ is computed using the separation ansatz $\alpha(x,y)=\psi(\frac12(x+y))\tau(x-y)$. Our work provides a rigorous justification of these simplifications and, furthermore, generalizes WHH theory to situations where the range of the two-body interaction cannot be neglected. While the emphasis in this paper was on temperatures close to $T_c$, we also developed tools applicable to the full temperature range. 

Previous work on WHH theory in the physics literature studying the local- and the phase approximations include \cite{EL1} and \cite{GG,RSK}, respectively. We mention that our result is restricted to weak coupling superconductors without impurities; see, e.g., \cite{SchSch} for an extension of WHH theory to strong coupling superconductors.

Throughout this appendix we assume that $\mu>0$. Our main result is a rigorous derivation of the following formula for the slope of $B_{c2}(T)$ at $T=T_c$, which we obtain without using the simplifications of WHH theory,
\begin{equation} 
\label{ourslope} 
\left.\frac{dB_{c2}(T)}{dT}\right|_{T=T_c} = - \frac{\int_{\R^3}  t_*(p)^2\, \cosh^{-2}(\beta_c(p^2-\mu)/2)\, dp}{\int_{\R^3}  t_*(p)^2 \left( g_1(\beta_c(p^2-\mu)) + \frac 23 \beta_c p^2 g_2(\beta_c(p^2-\mu))\right) dp}
\end{equation} 
with the special function $g_1(z)$ and $g_2(z)$ given in \eqref{eq:auxiliary}, and where $t_*(p)$ is the momentum dependent gap (usually denoted as $\Delta(p)$ in the physics literature).

The WHH result for the slope of $B_{c2}(T)$ at $T=T_c$ is 
  \begin{equation} 
\label{slope1} 
\left.\frac{dB_{c2}(T)}{dT}\right|_{T=T_c} = - \frac{T_c}{\mu} \frac{6}{\gamma_1} 
\end{equation} 
with
$$
\gamma_1 := \int_{\mathbb{R}}g_2(z)\,dz = 7\zeta(3)/\pi^2 \approx 0.8526 \,;
$$
see Eq.\ (23) in \cite{La2} with $\langle v_\perp^2\rangle_{FS} = 2v_F^2/3$ and $\mu = v_F^2/4$. In the rest of this appendix we show how to obtain the WHH result in \eqref{slope1} from ours in \eqref{ourslope}, and how to derive corrections to this. 

Since  $t_*(p)$ only depends on $|p|$, it can be written as a function of $p^2$,
\begin{equation*} 
t_*(p) = \tilde t_*(p^2)\, . 
\end{equation*} 
This  allows to change integration variables in \eqref{ourslope} to $x=\beta_c(p^2-\mu)$ and to obtain
\begin{align*}
\left. \frac{dB_{c2}(T)}{dT}\right|_{T=T_c} = 
-\frac{ \int_{\mathbb{R}} G(x/(\beta_c\mu))\cosh^{-2}(x/2)\, dx}{ \int_{\mathbb{R}} G(x/(\beta_c\mu)) \left( g_1(x) + \frac 23 (x+\beta_c \mu) g_2(x)\right)dx}
\end{align*}
with the notation, $\Theta$ being the Heaviside function,
\begin{equation*}
G(z) := \frac{\tilde t_*((1+z)\mu)^2}{\tilde t_*(\mu)^2} \sqrt{1+z} \,\Theta(1+z) \,.
\end{equation*} 
We found it convenient to fix an arbitrary multiplicative constant so that $G(0)=1$.

For standard BCS superconductors, $G(z)$ is a smooth function in some interval containing $z=0$, and $1/(\beta_c\mu)$ is very small ($10^{-3}$ or so, typically). Thus it is appropriate to compute the integrals above by inserting the Taylor series of $G$,  
$$
 G(x/(\mu \beta_c)) = 1 + G'(0) \frac{x}{\beta_c\mu} + G''(0) \frac{x^2}{2(\beta_c\mu)^2} +\ldots 
$$
A simple computation gives the integral in the nominator,
\begin{equation*} 
 \int_{\mathbb{R}} G(x/(\beta_c\mu))\cosh^{-2}(x/2) \,dx = 4 + \frac{2\pi^2}{3} G''(0)(\beta_c\mu)^{-2} + O((\beta_c\mu)^{-4})\,. 
\end{equation*} 
The corresponding computation for the integral in the denominator is more complicated, since it contains a logarithmic term,
\begin{align*} 
& \int_{\mathbb{R}} G(x/(\beta_c\mu)) \left( g_1(x) + \frac 23 (x+\beta_c \mu) g_2(x)\right)dx \\
& \qquad = \frac23 \gamma_1 \beta_c\mu + \left[2G'(0)\ln(\gamma_2\beta_c\mu)+\frac23 G''(0)\right](\beta_c\mu)^{-1}  + o((\beta_c\mu)^{-1})
\end{align*} 
with
\begin{equation*} 
\gamma_2:= \exp\left(\frac23 - \int_0^\infty \ln(x)\frac{d}{dx}(x^2 g_1(x))dx \right) \approx 0.8124 \,. 
\end{equation*} 
Using this we obtain 
\begin{equation*} 
\left. \frac{dB_{c2}(T)}{dT}\right|_{T=T_c} =  -\frac{6}{\gamma_1\beta_c\mu}\Bigl( 1 -\frac3{\gamma_1}G'(0)\frac{\ln(\gamma_2\beta_c\mu)}{(\beta_c\mu)^{2}}  + \left( \frac{\pi^2}{6}  - \frac1{\gamma_1}\right)\frac{G''(0)}{(\beta_c\mu)^{2}} +o\left(\frac1{(\beta_c\mu)^2}\right) \Bigr),
\end{equation*} 
which yields \eqref{slope1} and provides corrections to this WHH result. 

We finally mention that the terms $G'(0)$ and $G''(0)$ can be evaluated using \cite{FHNS}
\begin{equation*} 
t_*(p) = c_{*} \left( \int_{\mathbb{R}^3} V(r)j_0(\sqrt{\mu}|r|)j_0(|p||r|)dr + o(1)\right)
\end{equation*}
with the spherical Bessel function $j_0(z)=\sin(z)/z$ and an irrelevant constant $c_*$. This formula holds in the weak coupling limit, that is, when $V$ is replaced by $\lambda V$ with a constant $\lambda\ll 1$. We mention that the limit $\lambda\ll 1$ is consistent with the limit $\beta_c\mu\gg 1$ which we performed before. (On the other hand, from a mathematical perspective it is not completely obvious that Theorem \ref{main} is applicable since the assumption $T_c>T_1>0$ is not satisfied uniformly in $\lambda$. We plan to address this in future work.)


\bibliographystyle{amsalpha}

\begin{thebibliography}{17}

\bibitem{AbSt} M. Abramowitz, I. A. Stegun, \textit{Handbook of mathematical functions with formulas, graphs, and mathematical tables}. Reprint of the 1972 edition. Dover Publications, New York, 1992.

\bibitem{CoNe} H. D. Cornean, G. Nenciu, \textit{On eigenfunction decay for two-dimensional magnetic Schr\"odinger operators}. Comm. Math. Phys. \textbf{192} (1998), no. 3, 671--685.

\bibitem{Er} A. Erd\'elyi, \textit{Asymptotic forms for Laguerre polynomials}. J. Indian Math. Soc. \textbf{24} (1960), 235--250.

\bibitem{FHNS} R. L. Frank, C. Hainzl, S. Naboko, R. Seiringer, {\it The critical temperature for the BCS equation at weak coupling}, J. Geom. Anal. {\bf 17} (2007), 559--568.

\bibitem{FHSS} R. L. Frank, C. Hainzl, R. Seiringer, J. P. Solovej, \textit{Microscopic derivation of Ginzburg--Landau theory}. J. Amer. Math. Soc. \textbf{25} (2012), no. 3, 667--713.

\bibitem{FHSS2} R. L. Frank, C. Hainzl, R. Seiringer, J. P. Solovej, \textit{The external field dependence of the BCS critical temperature}. Comm. Math. Phys. \textbf{342} (2016), no. 1, 189--216.

\bibitem{GG} L.W. Gruenberg, L. Gunther, {\em Effect of Orbital Quantization on the Critical Field of Type-II Superconductors}, Phys. Rev. \textbf{176} (1968), 606.

\bibitem{HHSS} C. Hainzl, E. Hamza, R. Seiringer, J. P. Solovej, \textit{The BCS functional for general pair interactions}. Comm. Math. Phys. \textbf{281} (2008), no. 2, 349--367.

\bibitem{HaSe} C. Hainzl, R. Seiringer, \textit{Critical temperature and energy gap in the BCS equation}. Phys. Rev. B \textbf{77} (2008), 184517.

\bibitem{HSreview} C.~Hainzl, R.~Seiringer, \textit{The Bardeen--Cooper--Schrieffer functional of superconductivity and its mathematical properties}. J. Math. Phys. \textbf{57} (2016), 021101.

\bibitem{HeWe} E. Helfand, R. Werthamer, \textit{Temperature and purity dependence of the superconducting critical field, $H_{c2}$. II}. Phys. Rev. \textbf{147} (1966), no. 1, 288--294.

\bibitem{EL1} E. Langmann, {\em Theory of the upper critical magnetic field without local approximation}, Physica C \textbf{159} (1989), 561.

\bibitem{La} E. Langmann, \textit{$B_{c2}(T)$ of anisotropic systems: Some explicit results}. Physica B \textbf{165--166} (1990), 1061.

\bibitem{La2} E. Langmann, \textit{On the upper critical field of anisotropic superconductors}. Physica C \textbf{173} (1991), 347.

\bibitem{Ne} G. Nenciu, \textit{On asymptotic perturbation theory for quantum mechanics: almost invariant subspaces and gauge invariant magnetic perturbation theory}. J. Math. Phys. \textbf{43} (2002), no. 3, 1273--1298.

\bibitem{RSK} C.T. Rieck, K. Scharnberg, R.A. Klemm, {\em Re-entrant superconductivity due to Landau level quantization?},  Physica C \textbf{170} (1990), 195.

\bibitem{SchSch} M. Schossmann, E. Schachinger, {\em Strong-coupling theory of the upper critical magnetic field $H_{c2}$}, Phys. Rev. B \textbf{33} (1986), 6123.

\bibitem{WeHeHo} N. R. Werthamer, E Helfand, P. C. Hohenberg, \textit{Temperature and purity dependence of the superconducting critical field, $H_{c2}$. III}. Phys. Rev. \textbf{147} (1966), 295.

\end{thebibliography}

\end{document}